\documentclass[journal, onecolumn, draftnocls,10pt]{IEEEtran}

\usepackage{setspace}

\usepackage[table,dvipsnames]{xcolor}
\usepackage{times}
\usepackage{epsfig}
\usepackage{amsmath}
\usepackage{amsfonts}
\usepackage{graphicx}
\usepackage{amssymb}
\usepackage{amsthm}
\usepackage{amstext}
\usepackage{latexsym}
\usepackage{color}
\usepackage{ifthen}
\usepackage{multirow}
\usepackage{verbatim}
\usepackage{array,tabularx}
\usepackage{arydshln}
\usepackage[mathscr]{euscript}
\usepackage{accents}
\usepackage{cite}
\usepackage{hhline}
\usepackage{caption}
\usepackage{subcaption}
\usepackage{enumerate}
\usepackage{xcolor}
\usepackage{mathtools}
\usepackage{url}
\usepackage{xparse}
\usepackage{makecell}
\usepackage{varwidth}
\usepackage{soul}
\usepackage[normalem]{ulem}
\usepackage{pgfplots}
\usepackage{subcaption}
\usepackage{multirow}
\usepackage{cancel}

\newcolumntype{C}[1]{>{\centering\let\newline\\\arraybackslash\hspace{0pt}}m{#1}}

\DeclarePairedDelimiter\ceil{\lceil}{\rceil}

\newcommand{\be}{\begin{equation}}
	\newcommand{\ee}{\end{equation}}
\newcommand{\bear}{\begin{eqnarray}}
	\newcommand{\eear}{\end{eqnarray}}
\newcommand{\bears}{\begin{eqnarray*}}
	\newcommand{\eears}{\end{eqnarray*}}
\newcommand{\bi}{\begin{itemize}}
	\newcommand{\ei}{\end{itemize}}
\newcommand{\ben}{\begin{enumerate}}
	\newcommand{\een}{\end{enumerate}}

\usepackage{tikz}
\usetikzlibrary{shapes.geometric}
\usetikzlibrary{shapes,snakes,patterns}
\usetikzlibrary{arrows,positioning,automata,calc}
\usetikzlibrary{plotmarks}

\definecolor{ogreen}{rgb}{0,0.5,0}
\definecolor{magenta}{rgb}{1.0, 0.11, 0.81}
\definecolor{mulberry}{rgb}{0.77, 0.29, 0.55}
\definecolor{xgray}{rgb}{0.5, 0.5, 0.5}
\definecolor{ao}{rgb}{0.0, 0.5, 0.0}
\definecolor{amber}{rgb}{1.0, 0.75, 0.0}
\definecolor{capri}{rgb}{0.0, 0.75, 1.0}
\definecolor{chocolate}{rgb}{0.91, 0.41, 0.17}

\newcommand{\blue}{\color{blue}}

\newcommand{\F}{\mathbb{F}}

\DeclareMathOperator{\rk}{rk}

\newtheorem{theorem}{Theorem}

\newtheorem{proposition}{Proposition}

\newtheorem{example}{Example}
\newtheorem{definition}{Definition}

\newtheorem{remark}{Remark}

\newcommand*{\rowstyle}[1]{% sets the style of the next row
	\gdef\@rowstyle{#1}%
	\leavevmode\@rowstyle
	\ignorespaces
}

\newcolumntype{=}{% resets the row style
	>{\gdef\@rowstyle{}\ignorespaces}%
}
\newcolumntype{+}{% adds the current row style to the next column
	>{\leavevmode\@rowstyle\ignorespaces}%
}

\newcommand{\Fq}{\ensuremath{\mathbb{F}_q}}

\newcommand{\Fqs}{\ensuremath{\mathbb{F}_{q^s}}}
\renewcommand{\vec}[1]{\ensuremath{{#1}}}
\DeclareMathOperator{\rank}{rk}

\newcommand{\mycode}[1]{\ensuremath{\mathcal{#1}}}
\newcommand{\fontmetric}[1]{\mathsf{#1}}

\newcommand{\codelinearRank}[1]{\ensuremath{[#1]_q^\fontmetric{R}}}
\newcommand{\Gabcode}[2]{\ensuremath{\mathcal{G}(#1,#2)}}

\newcommand{\X}{\vec{X}}
\newcommand{\Y}{\vec{Y}}

\newcommand{\awcomment}[1]{\color{black}{#1 }\color{black}}

\title{Private Information Retrieval over Random Linear Networks}

\begin{document}

\doublespacing

\author{\IEEEauthorblockN{Razane Tajeddine\IEEEauthorrefmark{1}, Antonia Wachter-Zeh\IEEEauthorrefmark{2}, Camilla Hollanti\IEEEauthorrefmark{1}}\\
\IEEEauthorblockA{\IEEEauthorrefmark{1} Department of Mathematics and Systems Analysis, 
		Aalto University School of Science, 
        Espoo, Finland\\
		Emails: \{razane.tajeddine, camilla.hollanti\}@aalto.fi}\\
        		\IEEEauthorblockA{\IEEEauthorrefmark{2} Institute for Communications Engineering, 
		Technical University of Munich,  Germany\\
	Email: antonia.wachter-zeh@tum.de}\thanks{The work of R. Tajeddine and C. Hollanti was supported by the Academy of Finland,
under Grants No. 276031, 282938, and 303819, and by the Technical University of Munich – Institute for Advanced Study, funded by the German Excellence
Initiative and the EU 7th Framework Programme under Grant Agreement No. 291763, via a Hans Fischer Fellowship. The work of A. Wachter-Zeh was supported by the Technical University of Munich – Institute for
Advanced Study, funded by the German Excellence Initiative and European Union 7th Framework Programme under Grant Agreement No. 291763 and the German Research Foundation (Deutsche Forschungsgemeinschaft, DFG) under Grant No. WA3907/1-1.}}

\maketitle

%\doublespacing%to ease camilla's reading!

\begin{abstract}

In this paper, the problem of providing privacy to users requesting data over a network from a distributed storage system (DSS) is considered. The DSS, which is considered as the multi-terminal destination of the network from the user's perspective, is encoded by a maximum rank distance (MRD) code to store the data on these multiple servers. A private information retrieval (PIR) scheme ensures that a user can request a file without revealing any information on which file is being requested to any of the servers. In this paper, a novel PIR scheme is proposed, allowing the user to recover a file from a storage system with low communication cost, while allowing some servers in the system to collude in the quest of revealing the identity of the requested file. The network is modeled as a random linear network, \emph{i.e.}, all nodes of the network forward random (unknown) linear combinations of incoming packets. Both error-free and erroneous random linear networks are considered.

\end{abstract}

\section{Introduction}

Privacy is a major concern for Internet or network users. Whenever a user downloads a file from a server, they reveal their interest in the requested file. Private information retrieval (PIR) allows a user to hide the identity of their requested file from the servers. PIR was first introduced in \cite{chor1998private}, where the data is assumed to be replicated on multiple servers and the user is able to retrieve the file she wants privately. It was also proved that if the data is stored on a single server, then the only way to achieve PIR in the \emph{information theoretic} sense, \emph{i.e.}, privacy with no restrictions on the computational power of the server, is to download all the files.

A large body of literature on PIR appeared after its introduction, mostly focusing on minimizing the communication cost. PIR schemes with subpolynomial communication cost were constructed for multiple servers in \cite{yekhanin2008towards,efremenko20123, beimel2002breaking} and later for two servers in \cite{dvir20142}, calculating  the communication cost as the sum of the upload and the download cost.

Later on, much work was done to construct PIR schemes aiming to reduce the download cost, assuming the upload cost to be negligible with respect to the download cost. This assumption is based on the fact that the size of the uploaded query vectors depends only on the \emph{number} of files in the system, while the size of the downloaded response vectors depends on the \emph{size} of the files. More precisely, for a single sub-query, the query vector to a node consists of $m$ symbols in $\Fq$, where $m$ is the number of files and $\Fq$ is the $q$-element finite field, while the response vector from one node is $1$ symbol in $\Fq^w$, where $w$ is the length of the file. In distributed storage systems (DSSs), and in the information-theoretic reformulation of this problem~\cite{chan2014private}, the size of the files is assumed to be arbitrarily large, thus making the number of the files negligible with respect to the size of the files, \emph{i.e.}, $w$ is much larger than $m$. Therefore, the download cost dominates the total communication cost. These assumptions reflect practical scenarios such as storing and retrieving videos.

In more recent work on PIR, there is growing interest in studying PIR for coded data. In \cite{shah2014one}, it was shown that when the data is coded on an exponentially large number of servers, only one extra downloaded bit is needed to achieve privacy. The same low download complexity with a linear number of servers was achieved in \cite{blackburn2017pir}. Furthermore, a method to transform any linear replication-based PIR scheme into a scheme for coded data while minimizing the storage overhead was proposed in \cite{fazeli2015pir}.
In addition to various constructions, the fundamental limits on the download cost of PIR schemes have been characterized for replicated data in \cite{sun2016capacity} and for coded data in \cite{banawan2018capacity}.

The servers can be either considered to not communicate with each other and are therefore only aware of the query they received from the user, referred to as a \emph{non-colluding} system, or, it can be assumed that any set of maximum $t$ servers are communicating in an effort to figure out the requested file's identity, referred to as a \emph{colluding} system. We refer to the latter as \emph{$t$-PIR}, meaning that any subset of up to $t$ servers is allowed to collude. The capacity for $t$-PIR has been established in \cite{sun2016capacity}. For coded $t$-PIR  the capacity is still unknown, but a first scheme for $t$-PIR on coded servers employing a maximum distance separable (MDS) code was described in \cite{tajeddine2018privatej}. In \cite{freij2016private}, this scheme was extended to a wider set of parameters and an algebraic framework was established, also resulting in a conjecture for the coded $t$-PIR capacity, which was  disproved for a two-file ($m=2$) system in \cite{sun2017private}. Asymptotically ($m\rightarrow \infty)$, the conjectured capacity \cite{freij2016private} matches the capacity of \emph{symmetric} PIR \cite{wang2017symmetric,sun2016scapacity,wang2017secure1,wang2017secure}, %\cami{Razan, please add refs to Skoglund's and Jafar's SPIR papers and maybe a short explanation on linear vs non-linear, (un)coded/(non)colluded.} 
where the user can exclusively decode only the requested file. The work in \cite{freij2016private} was later extended to non-MDS codes in \cite{freij2018t}, and arbitrary linear codes were also considered in \cite{kumar2017achieving}.  {Symmetric PIR schemes over storage systems with Byzantine and colluding servers are discussed \cite{wang2018pir}. In \cite{banawan2019capacity}, the capacity of PIR schemes over a replicated storage systems with Byzantine and colluding servers was found. Noisy PIR over a storage system where the messages are replicated on multiple servers was discussed in \cite{banawan2018noisy}.}

In the present work, we construct a PIR scheme for a network with colluding nodes, where the data is encoded using a {maximum rank distance (MRD) code} \cite{silva2008rank}. The user sends the queries over a network of nodes to the destination nodes, which constitute a DSS, and the DSS servers send their responses over the random network to the user%\awcomment{[I think our approach allows that the network changes over time, i.e., that it's not the same network. Otherwise, a reviewer might wonder, why we don't make use of this...]}
, such that the identity of the requested file remains secret. %\cami{Should we say already here smth shortly about what we assume about the network? (Random, etc.)} \awcomment{[Yes. I wrote something, please check.]}

We model the network as in random linear network coding (RLNC), a concept that was introduced after the seminal observation \cite{Ahlswede_NetworkInformationFlow_2000} that forwarding linear combinations of the incoming packets at each node (called \emph{network coding}) instead of just forwarding packets (called \emph{routing}), increases the throughput. In RLNC \cite{koetter2008coding}, the structure of the network is unknown to the sink(s) and receiver(s) and might even change from time to time. Each node of the network forwards a random linear combination of its incoming packets and the goal of a certain receiver is to recover some or all of the transmitted packets. When the packets are modeled as vectors over a finite field and we see the set of transmitted packets as a linear subspace, then, in the error-free case, the received packets form a subspace of the transmitted packets. Also, one injected erroneous packet might propagate widely in the network and destroy all of the received packets, but the received subspace still has a large intersection with the transmitted one. This observation by K\"otter and Kschischang \cite{koetter2008coding}  motivated the use of \emph{subspace codes} for error-correction in RLNC. \emph{Lifted MRD codes} are special subspace codes that were proposed for RLNC in \cite{silva2008rank} and are efficiently decodable~\cite{RichterPlass_DecodingRankCodes_2004,WachterAfanSido-FastDecGabidulin_DCC_journ,PuchingerWachterzeh-ISIT2016}. In this paper, we use MRD codes as the query and storage codes and lift the code before transmitting the query/response over the network to be able to guarantee PIR while coping with errors.

\emph{Contributions:} The main contributions of this paper are as follows:
\begin{itemize}
\item To the best of the authors' knowledge, PIR schemes over a random linear network with data encoded using MRD codes is considered for the first time in this paper. We take this step towards a more practical model, compared to the earlier literature that assumes direct links between the user and the storage system when retrieving a file.
%\item This work assumes that the queries are sent over a noisy channel, which has not been considered before to the best of the authors' knowledge.\cami{I repeat my comment. Maybe remove this bullet, since in RLNC noisy channels have been considered, and above we already say that PIR here is new in this context, and below we say we also consider noisy channels so this bullet is both misleading and redundant imo. }
\item Two PIR schemes are given, one assuming an error-free channel, and another assuming a network with errors {, where the errors are considered in the uplink channel, as well as the downlink. To the best of the authors' knowledge uplink errors were not considered in any preceding work on PIR.}
\item The achieved PIR rate for an error-free network achieves the asymptotic PIR capacity when the field size is sufficiently large. For the network with errors, the scheme achieves the conjectured asymptotic PIR capacity given in \cite{tajeddine2018private}.
\end{itemize}

\section{Preliminaries}
%\awcomment{[A lot of things need to be properly introduced and notated, so I added Subsections A and B.]}

%\awcomment{[Some comments on notation: can we choose a common notation for vectors and matrices? I definitely don't like using mathcal, as it is usually used for sets. I suggest either to use nothing, i.e., $A$ or $\mathbf{A}$. I created a macro for this: $\vec{A}$]}

\subsection{Notation}
The following table provides an overview of the nomenclature used in this paper.
\begin{table}[htb]
\captionof*{table} {NOMENCLATURE} \label{tab:title} 
\vspace{-1em}
\begin{center}
\begin{tabular}{|c|p{7cm}|}
\hline
$m$ & Number of files \\\hline
$l$ & Number of servers\\\hline
$n$ & Number of sub-servers in an $\codelinearRank{n,k,d}$ MRD code\\\hline
$k$ & Dimension of the codeword in an $\codelinearRank{n,k,d}$ MRD code\\\hline
$d$ & Minimum distance of an $\codelinearRank{n,k,d}$ code\\\hline
$t$ & Number of colluding nodes \\\hline
$\rho=\frac{n}{l}$ & Columns stored on a server\\\hline
$\beta$ & Number of subdivisions / stripes \\\hline
$\X$ & Set of files stored in the system \\\hline
$X^f$ & File requested by the user \\\hline
%$\diag{AB}$ & Vector composed of the diagonal elements of $AB$ \\\hline
$R_{\textup{\textsf{PIR}}}$ & PIR rate\\\hline
$\Y_j$ & Data stored on server $j$ \\\hline
$Q^{f,(i)}$ & Query matrix sent in round $i$ to retrieve file $X^f$\\\hline
$Q^{f,(i)}_{j,rec}$ & Query matrix received by server $j$ in round $i$ to retrieve file $X^f$\\\hline
$R^{f,(i)}_{j,rec}$ & Received response matrix from server $j$ in round $i$ \\\hline
$sk$ & File size \\\hline
$\epsilon$ & Number of errors introduced in the network\\\hline
$\tau$ & Number of erasures in the network\\\hline
%$\gamma$ & Number of deviations in the network\\\hline
$\mu$ & $m\beta/s$\\\hline
%$q_A$ & Field size of the entries of the network matrices $A$ and $A'$\\\hline
\end{tabular}
\end{center}
\vspace{-1em}
\end{table}

%\cami{This is surely not all notation we have..? For instance the notation introduced in system model section could be added here as well to ease the reading later. Add here $\beta,\mu,\rho$...}

Let $q$ be a power of a prime and let
$\Fq$ denote the finite field of order $q$ and $\Fqs$ its extension field of order $q^s$. 
We use $\Fq^{s \times n}$ to denote the set of all $s\times n$ matrices over $\Fq$ and $\Fqs^n =\Fqs^{1 \times n}$ for the set of all row vectors of length $n$ over $\Fqs$.

Rows and columns of $s\times n$-matrices are indexed by $1,\dots, s$ %\cami{isn't this trivial, no need to explain...?} \awcomment{I think it's worth mentioning whether we index starting from 0 or from 1.}
and $1,\dots, n$, respectively. Denote the set of integers $[a,b] = \{i: a \leq i \leq b\}$. By $\rk_q(\vec{A})$ and $\rk_{q^s}(\vec{A})$, we denote the rank of a matrix $\vec{A}$ over $\Fq$, respectively $\Fqs$. 

%Each file will be subdivided into $\beta$ parts.

% Let $\vec{e}^{f}$ be the matrix of size $\beta\times m\beta$, such that
% $$\vec{e}^{f} = \left(\begin{array}{c:c:c}
% 	\mathbf{0}_{\beta(f-1)} & I_{\beta\times \beta} &\mathbf{0}_{\beta(m-f)}
% 	\end{array}\right),$$
%     then, define matrix $E^{f,(i)}$
%     \begin{equation}\label{eq:ef}
%     \vec{E}^{f,(i)} = \left(\begin{array}{c}
% 	\mathbf{0}_{i-1\times m\beta}\\
%     \vec{e}^{f}\\
%     \mathbf{0}_{n-\beta-i+1\times m\beta}
% 	\end{array}\right).\end{equation}

%Let \awcomment{$i_n \in [0,n]$ be such that} $$i_n = \begin{cases}n \quad\quad\quad\quad\quad~ \text{ if } i = n\\ i ~(\text{mod } n) \quad~~ \text{ otherwise.} \end{cases}$$

The following proposition will be useful \awcomment{in} this work:

\begin{proposition}(Probability of full rank, \cite{Ho_fullrank})\label{prop:prob}
Let $A\in\mathbb{F}_q^{\kappa\times\kappa}$ be a square matrix with elements chosen uniformly at random from $\mathbb{F}_q$. Then 
$$P=\mathbb{P}(\rk_{q}(A)=\kappa)\geq \left(1-\frac{1}{q}\right)^\kappa\,.$$
\end{proposition}

%\awcomment{That way, $i_n = 0$ if it was $n$ before. I guess you mean it should become $1$?}
%\cami{The above makes no sense to me. What is our goal here? Isn't the above equivalent to saying that $i_n=i\ \mod\ n$, except when $i=n$, when we define it to remain $n$ (instead of zero)?}\rt{Yes, that is exactly what I want.}

%Denote by $\Mooremat{s}{q}{\vec{a}} \in \Fqs^{m \times n}$ the $m \times n$ Moore matrix for a vector $\vec{a} = (a_1,a_2,\dots,a_n) \in \Fqs^n$ of length $n$, i.e.:
%\begin{equation*}
%\Mooremat{m}{q}{\vec{a}} = \MoormatExplicit{a}{m}{n}.
%\end{equation*}
%If $a_1, a_2,\dots$, $a_{n}\in \Fqs$ are linearly independent over $\Fq$, then $\rk_{q^s}(\Mooremat{m}{q}{\vec{a}})=\min\{m,n\}$, cf.~\cite[Lemma 3.15]{Lidl-Niederreiter:FF1996}.
\subsection{Rank-Metric Codes and Gabidulin Codes}
By utilizing the vector space isomorphism $\Fqs \cong \Fq^s$, there is a bijective map from vectors $a \in \Fqs^n$ to matrices $A {\triangleq \Phi(\vec{a})} \in \Fq^{s \times n}$.
The rank distance between \vec{a} and \vec{b} $\in \Fqs^n$ is the rank of the difference of the two matrix representations:
\begin{equation*}
d_{\textup{R}}(\vec{a},\vec{b})\triangleq \rk_q(\vec{a}-\vec{b}) = \rk_q(\vec{A}-\vec{B}).
\end{equation*}
An $\codelinearRank{n,k,d}$ %\awcomment{[In the nomenclature table it's $(n,k,d)$. Both is fine with me, but it should be unified.]}\rt{Ok, changed to this notation.}
 code \mycode{C} over $\Fqs$ is a linear rank-metric code, \emph{i.e.}, it is a linear subspace of $\Fqs^n$ of dimension $k$ and minimum rank distance $d$.
For linear codes with $n \leq s$, the Singleton-type upper bound \cite{Delsarte_1978,Gabidulin_TheoryOfCodes_1985} implies that $d \leq n-k+1$.
If $d=n-k+1$, the code is called a \emph{maximum rank distance} (MRD) code.

Gabidulin codes \cite{Gabidulin_TheoryOfCodes_1985} are a special class of rank-metric codes.
%\begin{definition}[Gabidulin Code \cite{Gabidulin_TheoryOfCodes_1985}]
A linear Gabidulin code $\Gabcode{n}{k}$  over $\Fqs$ of length $n \leq s$ 
and dimension $k$ is defined by evaluating degree-restricted linearized polynomials:
\begin{align*}
\Gabcode{n}{k} = \{&(f(\alpha_0), f(\alpha_1), \dots, f(\alpha_{n-1})) :\\ &f(z) = f_0z^{q^0} + f_1z^{q^1}+ \dots + f_{k-1}z^{q^{k-1}}\},
%\Mooremat{k}{q}{g_1,g_2,\dots,g_n}
%	\MoormatExplicit{g}{k}{n},
\end{align*}
where $\alpha_0,\alpha_1,\dots,\alpha_{n-1} \in \Fqs^n$ have to be linearly independent over $\Fq$.
%\awcomment{[I removed the generator matrix and replaced it by the evalutaion as this is what we need later.]}
Gabidulin codes are MRD codes, \emph{i.e.}, $d=n-k+1$, cf.~\cite{Gabidulin_TheoryOfCodes_1985,Roth_RankCodes_1991}.

%\rt{We can view a Gabidulin code as an evaluation code, such that the generator matrix can be defined by $k$ functions $f_j(z)=z^{q^j}, j=0,\cdots, k-1$ evaluated on $n$ evaluation points, $\alpha_i,i=1,\cdots,n$. Hence, $\mathbf G_{\mycode{G}} = Eval[\{f_0(z),\cdots, f_{k-1}(z)\}, (\alpha_i, i=1,\cdots,n)]$. }

\subsection{System Model}\label{sec:model}

%Assume a network with multiple servers and a user, where the data $\X_{m\times k}$ is encoded using a $[n,k]$ Gabidulin generator matrix over $\mathbb{F}_{q^s}$, $G_{Gab}$. $$\Y_{m\times n}=\X G_{Gab}.$$
% * <rstajeddine@gmail.com> 2018-10-21T12:58:21.167Z:
%
% > [n,k
%
% ^.

%\cami{I put the picture at the bottom of the page,  it looked weird with just one sentence and then figure.}
\begin{figure}[b!]
\centering
\includegraphics[scale=0.6]{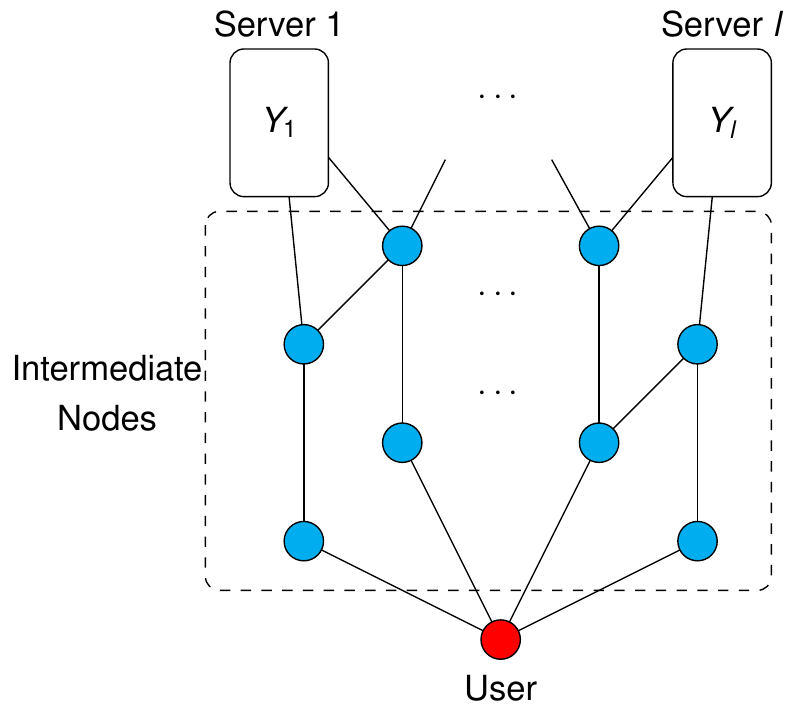}
\caption{Network Model: The user is connected to the servers via intermediate nodes. When data is sent from the user to the servers, or from server to user, the packets go through a random network using the intermediate nodes. {The intermediate nodes compute random linear combinations of the received packets and send them forward.  The intermediate network nodes are considered to be non-colluding and to simply relay data. We also assume the typical RLNC model in which all packets of one ``shot'' arrive at the same time, and therefore no memory or data handling is needed. }
}\label{fig:Network}
\end{figure}

%\cami{Here, a picture of a network would be nice, showing that the DSS is the set of destination nodes.}\awcomment{Agree.}

Assume a network with multiple servers, where the stored data $\vec{X}\in\Fqs^{m\times k}$ is formed by $m$ files $X^1,\ldots,X^m\in \mathbb{F}_q^{s\times k}$, where each file $X^i$ is subdivided into $\beta$ stripes, where each stripe is assumed to be on a separate row, such that 
\[
	\vec{X}={\Phi}^{-1}\left( 
	\begin{array}{c}
	X^1\\
	X^2\\
	\vdots\\
	X^m
	\end{array}\right) \in \Fqs^{m\beta \times k},\]
where $\Phi^{-1}$ denotes the one-to-one mapping from $s \times k$ matrices in $\Fq$ to vectors in $\Fqs^k$ (\emph{i.e.}, the inverse mapping of $\Phi$).
	 Each row of the data $\vec{X} \in \Fqs^{m\beta \times k}$, \emph{i.e.}, each stripe, is encoded with a $\Gabcode{n}{k}$ code over $\Fqs$ such that we obtain $m\beta$ Gabidulin codewords, forming %contained in 
	 the rows of a matrix $\vec{Y} \in \Fqs^{m\beta \times n}$.
	 %matrix representation of encoding $\Phi^{-1}(X^j)$.

%\awcomment{[removed the generator matrix as we do everything by evaluation now; please check the definition of $Y_j$ as this is used later... I think something is wrong as the sizes don't work out later...]		}

%The data $\vec{X}$ is encoded using the generator matrix $\mathbf G_{\mycode{G}}$ of a $\Gabcode{n}{k}$ code over $\Fqs$: 
%\begin{equation*}
%\vec{Y} = \vec{X}\cdot \mathbf G_{\mycode{G}},
%\end{equation*}
%where $\vec{Y} \in \Fqs^{m \times n}$.

%\awcomment{%[new:]
%Assume a network with multiple servers and a user, where the data $\vec{X} \in \Fq^{s \times k}$ is encoded using the generator matrix $\mathbf G_{\mycode{G}}$ of a $\Gabcode{n}{k}$ code over $\Fqs$: 
%\begin{equation*}
%\vec{Y} = \vec{X}\cdot \mathbf G_{\mycode{G}},
%\end{equation*}
%where $\vec{Y} \in \Fqs^{m \times n}$.
%	}

%\rt{We divide the matrix $\vec{Y}$ into $l$ blocks of columns, which we call \emph{column blocks}.}
%Column blocks \awcomment{[define what column blocks are...]} of the matrix $\vec{Y}$ are stored on separate servers. We consider a random network, and store each block of the matrix $\Y$ on a server. 
%We assume in this setting that there are $l$ servers, such that $n$ is divisible by $l$, and that the servers are honest but curious, and $t$ of them collude (share the queries they receive with one another). 

	Assume that {the data is stored on} $l$ servers where $n$ is divisible by $l$. 
	We divide the matrix $\Y$ into $l$ blocks of $\rho \triangleq n/l$ columns and store each such block on a separate server. {By $Y_j \in \Fqs^{m\beta\times\rho}$, we denote the $j^{th}$ block of $\vec{Y}$, stored on server $j$.}
	The servers are assumed to be honest-but-curious, and at most $t$ of them may collude, \emph{i.e.}, share the queries that they receive with each other. Fig.~\ref{fig:Network} illustrates the system model.

Note that according to the setup of this problem, each server is storing $\rho$ columns of $\Y$, which means the privacy constraint should take into account that each server acts as $\rho$ colluding servers. For simplicity, we view the system as a system of $n$ sub-servers with up to $\rho t$ collusions. Technically, the collusion in this setting is not full $\rho t$ collusion as not any $\rho t$ sub-servers can collude. The problem of constructing PIR schemes with certain collusion patterns is considered in \cite{tajeddine2017private}. 

Throughout this paper, we will use superscripts to denote the file index, and subscripts to denote the server index.

\subsection{PIR Scheme and PIR Rate}\label{sec:scheme}
%\awcomment{[This section needs to get more formal. Explain what is the storage code, the PIR code, ...]}

{We assume that the data is stored on $l$ servers as described earlier. PIR allows the user to retrieve a file $X^f \in \Fq^{s \times k}$, $f \in \{1,\dots,m\}$, from the database without revealing the identity $f$ of the file to any of the servers. A \emph{linear PIR scheme} is a scheme over $\mathbb{F}_q$, consisting of the following stages:

\begin{itemize}
	\item Query stage: When the user wants to retrieve a file $X^f$, she sends a  query matrix $\vec{Q}^f\in \Fq^{n\times m\beta}$ %is chosen uniformly at random and 
	 to the servers. {%By $\vec{Q}_j^f \in \Fq^{m}$, we denote row $j$ of the matrix $\vec{Q}^f$. 
    We divide the query matrix into $l$ blocks of $\rho$ rows each, where block $j$ is denoted by $\vec{Q}_j^f \in \Fq^{\rho\times m\beta}$, and is sent to server $j$.} %\emph{P.S. We can view each server as a set of $\rho$ colluding sub-servers, as described earlier, thus viewing the storage system as a set of $n$ sub-servers, such that the query vector $\vec{Q}_j^f$ is sent to sub-server $j$.}} % row $j$ of $\vec{Q}^f$ $\vec{Q}_j^f \in \Fq^{m}$ is sent to server $j$.
%	\awcomment{[Say how  $\vec{Q}_j^f$ is defined]}
	\item %\cami{I think using $\cdot$ for matrix product is confusing, can we just denote $AB$, instead of $A\cdot B$?} 
    Response stage: Upon receiving the query, server $j$ responds to the user by projecting its stored data onto the query matrix. The response matrix, $\vec{R}_j^f$, is the lifted diagonal matrix consisting of the diagonal elements of the product of the matrices $\vec{Q}_j^f$ and $Y_j$, defined formally in section~\ref{sec:noerrnet}.
    %Response stage: Upon receiving the query, server $j$ responds to the user by projecting its stored data onto the query matrix in the following way: %\rt{Server $j$ multiplies row $\iota$ of the query block, $\vec{Q}_j^f$, it receives onto column $\iota$ of the block $Y_j$ stored}: %\rt{If we view $Y_j$ as a vector of length $m$ in $\Fqs$, we can write the response vector as:}
    %\[\vec{R}_j^f = \diag{\vec{Q}_j^f}{Y_j} \in {\Fqs^\rho},\]
    %where $\diag{A}{B}$ denotes the vector composed of the diagonal elements of the product of the matrices $A$ and $B$. 
    
%	\awcomment{where the triangle brackets (how to call them????) denote the scalar product built column-wise. [Probably I defined $Y_j$ wrong above, the matrix sizes don't work out.]}
\end{itemize}}

We next define the PIR rate which is the ratio of the number of downloaded symbols from the network if no privacy is required and the total number of downloaded symbols with privacy assumption, \emph{i.e.}, the number of downloaded information symbols and the total number of the downloaded symbols. The number of information symbols is the size of the file, \emph{i.e.}, $sk$, and we assume the user downloads $\delta$ symbols from each server. %\awcomment{[moved the text before definition and left only formulas in definition.]}
\begin{definition}[PIR rate] 
	The \emph{PIR rate  $R_{\textup{\textsf{PIR}}}$} %\awcomment{We could use a more beautiful subscript like $R_{\textup{\textsf{PIR}}}$ or similar. What do you think?} \cami{yes, razan can you change this.} 
	is \[R_{\textup{\textsf{PIR}}}=\frac{sk}{l\delta}.\]
%	\awcomment{Can we do that more technical? I.e., introduce letters for "downloaded symbols" and so on. Also, we haven't talked about lifting at all yet, so I would not mention this yet.}
\end{definition}

\subsection{The Star Product}\label{subsec:star-prod}

%\awcomment{[Maybe add something about the original star product scheme and why we chose to define it like this in our setting...]}
The star product scheme for constructing a PIR scheme on a storage system that is encoded with a {Generalized Reed--Solomon} (GRS) code was first introduced in \cite{freij2016private}. In that work, the star product between two codewords is the element-by-element product between {the symbols of two} codewords. The reason for choosing the star product as such is that the star product between two GRS codes {results in another} GRS code. In this work, however, we define the star product differently, as the storage system is encoded using a Gabidulin code, namely, such that  the star product between two Gabidulin codes is also a Gabidulin code.

Let $\alpha_1,\dots, \alpha_n \in \Fqs$ be linearly independent over $\Fq$ which implies $n \leq s$. %and let $\{ f_0, \dots, f_{k-1} \}$ and $\{ g_0, \dots, g_{t-1} \}$ denote two sets of functions. %\awcomment{[I guess it would be easier to understand if we base this definition on $f(x)$, $g(x)$ being degree-restricted linearized polynomials.]}.
Let the storage code $\mycode{C}$ be a $\Gabcode{n}{k}$ Gabidulin code, \emph{i.e.}:
\begin{align*}
\mycode{C} = \{ &(f(\alpha_0),f(\alpha_1),\dots,f(\alpha_{n-1})):\\
&f(z) = f_0 z^{q^0} + f_1 z^{q^1} + \dots + f_{k-1}z^{q^{k-1}}\}.\end{align*}

%where the $\alpha_i$ are independent.

%Let the  \rt{generator matrix} of the storage code be:
%\begin{align*}
%G_C &= \Eval\left[\{f_i\},(\alpha_j)\right] \\
%	&= \begin{bmatrix}
%  f_1(\alpha_1) & \cdots & f_1(\alpha_n)\\
%  \vdots & \ddots & \vdots\\
%  f_k(\alpha_1) & \cdots & f_k(\alpha_n)
%  \end{bmatrix}.\end{align*}

Now, let the query code $\mycode{D}$ be a $\Gabcode{n}{t}$ Gabidulin code, where
\begin{align*}\mycode{D} = \{&(g(\alpha_0),g(\alpha_1),\dots,g(\alpha_{n-1})):\\&g(z) = g_0 z^{q^0} + g_1 z^{q^1} + \dots + g_{t-1}z^{q^{t-1}}\}.\end{align*}

%\begin{align*}
%G_D &= \Eval\left[\{g_i\},(\alpha_j)\right] \\
%	&= \begin{bmatrix}
%g_1(\alpha_1) & \cdots & g_1(\alpha_n)\\
%\vdots & \ddots & \vdots\\
%g_t(\alpha_1) & \cdots & g_t(\alpha_n)
%\end{bmatrix}.\end{align*}

The \emph{star product} of $\mycode{C}$ and $\mycode{D}$ is then defined by:
\begin{align*}
\mycode{C}\star \mycode{D} = \{ &(h(\alpha_0),h(\alpha_1),\dots,h(\alpha_{n-1})) :\\&h(z) = f(g(z)),\\ 
&f(z) = f_0 z^{q^0} + f_1 z^{q^1} + \dots + f_{k-1}z^{q^{k-1}},\\ 
&g(z) = g_0 z^{q^0} + g_1 z^{q^1} + \dots + g_{t-1}z^{q^{t-1}}\},
\end{align*}
which is a $\Gabcode{n}{k+t-1}$ Gabidulin code.

%$$G_{C\star D_{k+t-1\times n}} = \Eval\left[\{f_i(g_j)\},(\alpha_k)\right].$$

%\awcomment{[Are all this matrices the code or just a generator matrix of the code? If the first, why call it $G$, if the second, write that these are generator matrices.]}

\begin{example}\label{starprod}
	Consider a storage system over $\mathbb{F}_{2^5}$ with primitive element $\alpha${ and primitive polynomial $z^5+z^2+1$}. Let $\mycode{C}$ be a $\Gabcode{5}{3}$ Gabidulin code defined by:
	\begin{align*}
	\mycode{C} = \{&(f(1),f(\alpha),f(\alpha^2),f(\alpha^3), f(\alpha^4)):\\&f(z) = f_0 z + f_1 z^{2} + f_{2}z^4\},\end{align*}
	%\awcomment{[This is no Gabidulin code! The code locators have to be linearly independent over $\Fq$. Here, $\alpha^2 + 1$ is clearly the sum of $\alpha^2$ and $1$. Over $\mathbb{F}_{2^3}$, there are at most $3$ linearly independent elements! To get a Gabidulin code of length $5$, you need at least $\mathbb{F}_{2^5}$...]}
	and $\mycode{D}$ be a $\Gabcode{5}{2}$ Gabidulin code defined by:
	\begin{align*}\mycode{D} = \{ &(g(1),g(\alpha),g(\alpha^2),g(\alpha^3), g(\alpha^4)):\\&g(z) = g_0 z + g_1 z^2\}.\end{align*}	
	%with $g_1(z)=z$, and $g_2(z)=z^2$. \awcomment{[I don't get how $f_1$, ... define a Gabidulin code... Can you somehow explain that in the definition of Gabidulin codes that I added?]}
%
%$$G_C = \Eval\left[\{f_i(z)\},(\alpha_j)\right] = \Eval\left[\{z,z^2,z^4\},(1, \alpha, \alpha+1, \alpha^2, \alpha^2+1 )\right],$$
%and $$G_D = \Eval\left[\{g_i(z)\},(\alpha_j)\right] = \Eval\left[\{z,z^2\},(1, \alpha, \alpha+1, \alpha^2, \alpha^2+1 )\right].$$
Now the code $\mycode{C}\star \mycode{D}$ is a $\Gabcode{5}{4}$ Gabidulin code defined by: 
	\begin{align*}
	\mycode{C}\star \mycode{D} = \big\{&(h(1),h(\alpha),h(\alpha^2),h(\alpha^3), h(\alpha^4)):\\
	 h(z) &= f(g(z))\\ 
	 &= f_0(g_0)z+(f_0(g_1)+f_1(g_0))z^2+\\
	 &~~~~(f_1(g_1)+f_2(g_0))z^4+f_2(g_1)z^8 \big\}.
	\end{align*}
	A generator matrix of the code $\mycode{C}\star \mycode{D}$ is given in Figure~\ref{tab:long-eq}.
	\begin{figure*}
	$$G_{C\star D} = \small\bordermatrix{
		& & & \cr
	{\blue{z^1}} & 1 & \alpha & \alpha^2 & \alpha^3 & \alpha^4 \cr
	{\blue{z^2}} & 1 & \alpha^2 & \alpha^4 & \alpha^3+\alpha & \alpha^3+\alpha^2+1 \cr
	{\blue{z^4}} & 1 & \alpha^4 & \alpha^3+\alpha^2+1 & \alpha^3+\alpha^2+\alpha & \alpha^4+\alpha^3+\alpha+1 \cr
	{\blue{z^8}} & 1 & \alpha^3+\alpha^2+1 & \alpha^4+\alpha^3+\alpha+1 & \alpha^4+\alpha^3+\alpha^2+\alpha & \alpha}
	.$$
	\protect\caption{Generator matrix of the code $\mycode{C}\star \mycode{D}$ in Example~\ref{starprod}.}
	\label{tab:long-eq}
    \end{figure*}
    %\awcomment{You should add the primitive polynomial that is used for the field definition, else it's not possible to check this matrix....}
  {Clearly, this definition of the star product works for any linear rank-metric code. }	
%	  {Thank you for changing this section! It's much clearer now!}
%In other words,
%
%\[
%G_C = \bordermatrix{
%& & & \cr
%{\blue{z^1}} & 1 & \alpha & \alpha+1 & \alpha^2 & \alpha^2+1 \cr
%{\blue{z^2}} & 1 & \alpha^2 & \alpha^2+1 & \alpha^2+\alpha & \alpha^2+\alpha+1 \cr
%{\blue{z^4}} & 1 & \alpha^2+\alpha & \alpha^2+\alpha+1 & \alpha & \alpha+1}
%,\]
%
%and 
%
%\[
%G_D = \bordermatrix{
%	& & & \cr
%	{\blue{z^1}} & 1 & \alpha & \alpha+1 & \alpha^2 & \alpha^2+1 \cr
%	{\blue{z^2}} & 1 & \alpha^2 & \alpha^2+1 & \alpha^2+\alpha & \alpha^2+\alpha+1}
%.\]
%
%Thus, $$C\star D = \Eval\left[\{z^1,z^2,z^4,z^8\},(1, \alpha, \alpha+1, \alpha^2, \alpha^2+1)\right].$$

%\[
%G_{C\star D} = \bordermatrix{
%	& & & \cr
%{\blue{z^1}} & 1 & \alpha & \alpha+1 & \alpha^2 & \alpha^2+1 \cr
%{\blue{z^2}} & 1 & \alpha^2 & \alpha^2+1 & \alpha^2+\alpha & \alpha^2+\alpha+1 \cr
%{\blue{z^4}} & 1 & \alpha^2+\alpha & \alpha^2+\alpha+1 & \alpha & \alpha+1 \cr
%{\blue{z^8}} & 1 & \alpha & \alpha+1 & \alpha^2 & \alpha^2+1}
%.\]

\end{example}

\subsection{Random Linear Network Coding}\label{subsec:random-network}
In this paper, data is transmitted over a random linear network, \emph{i.e.}, each node computes linear combinations of the received packets and forwards the linear combination, see \cite{koetter2008coding} for more details. 
%In this paper, data is transmitted over a network which behaves as in random linear network coding \cite{koetter2008coding}, \emph{i.e.}, each node performs linear combinations of the received packets and forwards this linear combination. 
Additionally, erroneous packets can be inserted on any edge of the network.
The random network channel is therefore modeled by
\begin{equation*}
\vec{R} = \vec{A}\vec{Q} + \vec{N}_e,
\end{equation*}
where the rows of $\vec{Q} \in \Fqs^{n \times m\beta}$ denote the packets that should be transmitted, $\vec{A} \in \F_{q_N}^{n\times n}$ denotes the channel matrix, $\vec{N}_e \in \F_{q_N}^{n \times m\beta}$ the overall error matrix and the rows of $\vec{R} \in \Fqs^{n \times 1}$ denote the packets at the receiver side. We assume, for simplicity, that the field size, $q_N$, of the entries of the network matrices is equal to $q^s$. This assumption is easy to generalize as long as the characteristics match, \emph{i.e.}, $\mathrm{char}(\F_{q_N})= \mathrm{char}(\F_{q^s})$.
%  {[Razan, can you please change $m,n,N$ according to what is used later in the paper...]}
%\cami{Do we really want to boldface matrices? Or even vectors?}  {No, not necessarily. I just made a macro which makes it it easy to change! change the macro to whatever you prefer!}

\begin{remark}
  {Throughout this paper, the results will not depend on the topology of the network. This is in line with the assumption of random linear network coding, see \emph{e.g.} the seminal papers \cite{koetter2008coding, silva2008rank}. Thus, our results hold for any network and the network structure does not have to be known, neither to the user nor to the servers. Even more, the structure is allowed to change during different uses of the network.}
\end{remark}

\section{PIR over an Error-Free Network}\label{sec:noerrnet}

%These colluding sets, however, are only considered to collude between each other and not with any of the other servers.

{\textbf{Query:}} 
%  {[Razan, I am completely confused about matrix sizes and fields here. Please overwork this and make clear if $U$ is encoded to one or to $m$ codewords and add the proper sizes and fields of the matrices here.]}
In the following, we assume that $s$ divides $m$.
We assume the user wants to retrieve the file $\vec{X}^f$, $f \in \{1,\dots,m\}$, while hiding the identity $f$ from the storage system as defined in Section~\ref{sec:model}. For this purpose, the files will be subdivided into $\beta=n-k-\rho t+1$ stripes, and a file will be retrieved in $k$ rounds of queries. For query $i$, the user will generate $t\rho$ random vectors of length $m\beta$, denoted by $\vec{u}_1^{(i)},\ldots, \vec{u}_{t\rho}^{(i)}$  $\in \Fq^{m\beta}$, such that \begin{equation}\label{eq:u}
\vec{U}^{(i)}=\left(\begin{array}{c}
\vec{u}^{(i)}_1\\
\vdots\\
\vec{u}^{(i)}_{t\rho}
\end{array}\right)^\top,\end{equation} where $\vec{U}^{(i)}\in\F_q^{m\beta\times t\rho}$. These vectors are then mapped onto $t\rho$ vectors of size $\mu = m\beta/s$ in $\Fqs$. %, which are then encoded using a $\Gabcode{n}{t\rho}$ Gabidulin code over $\Fqs$}, as shown in Figure~\ref{}. %to encode the mapping of $t\rho$ random vectors {of length $m$, denoted by $\vec{u}_1,\cdots, \vec{u}_{t\rho}$    {$\in \Fqs^m $ [correct??]}
%  {[Why in $\Fq$? Or, is $U$ first mapped to a vector $u \in \Fqm^{t\rho}$ and then encoded to a single Gabidulin codeword? But then the Gabidulin code is over $\Fqm$ not over $\Fqs$. ]}. 
The vectors are chosen uniformly at random so that the queries received by any $t\rho$ servers reveal no information about the index of the requested file.

% \begin{figure}
% \includegraphics[scale=0.7]{Mapping.pdf}
% \caption{Encoding Random Vectors into the Query}\label{fig:map}
% \end{figure}

\begin{figure*}[htb]
	\centering
\includegraphics[scale=0.55]{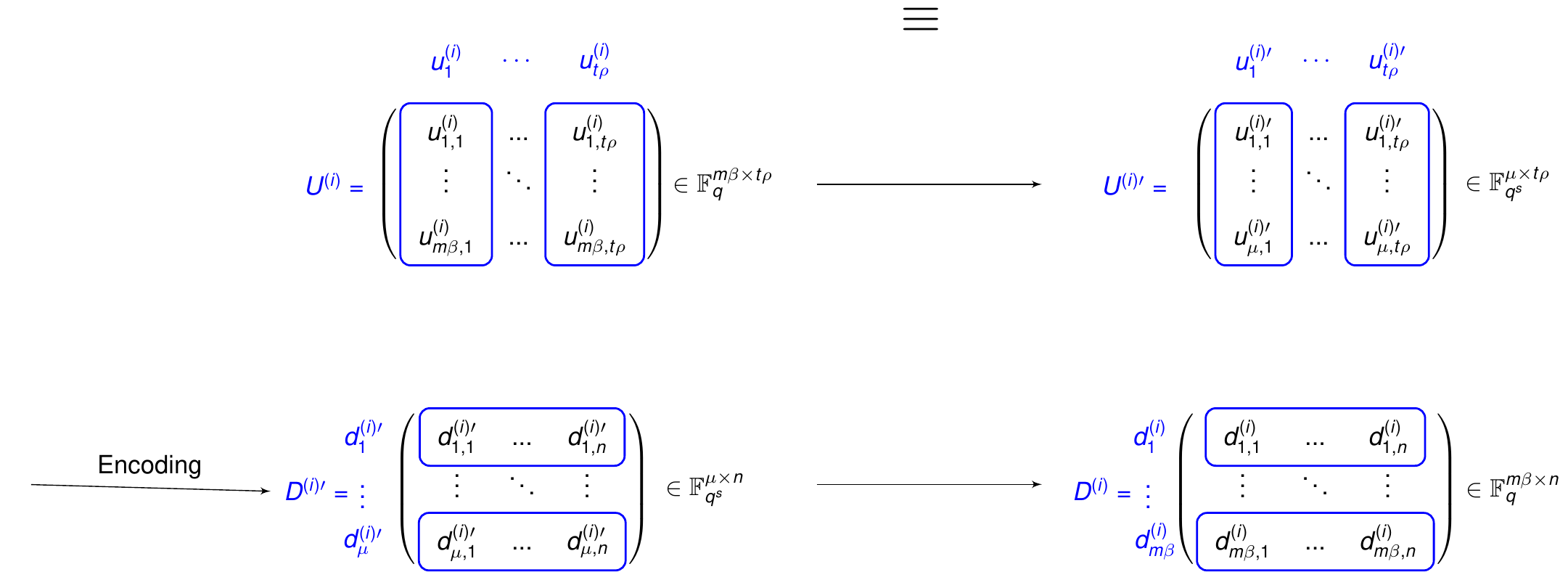}
\caption{Encoding random vectors into the query.  {$t\rho$ random vectors of length $m\beta$ are generated, then mapped onto vectors of size $m\beta/s$. Those vectors are then encoded into codewords of a $\Gabcode{n}{t\rho}$ Gabidulin code. In turn, those codewords are then mapped onto $m\beta$ codewords to form matrix $D^{(i)}$.}}\label{fig:map}
\end{figure*}

%  {[Say something why these vectors are random]}
The user then generates the matrix $\vec{D}^{(i)\prime} \in \Fqs^{\mu\times n}$ by encoding the $\mu$ rows of $U^{(i)\prime}$ to $\mu$ codewords of the $\Gabcode{n}{t\rho}$ Gabidulin code \mycode{D}. Those codewords are then mapped onto $m\beta$ codewords of \mycode{D} to form the matrix $\vec{D}^{(i)} \in \Fq^{m\beta\times n}$, as shown in Fig.~\ref{fig:map}.

Afterwards, the user forms the deterministic matrix $\vec{E}^{f,(i)}\in \Fq^{n\times m\beta}$ (see~\eqref{eq:ef}), which is used to retrieve the file $\vec{X}^f$. 
In one query round, the user can retrieve $\beta$ stripes of the requested file, as discussed in \cite{tajeddine2018private}. Therefore, the matrix $\vec{E}^{f,(i)}$ is chosen such that this matrix adds a $1$ to the randomly generated matrix $\vec{D}^{(i)}$ in $\beta$ positions of the required file $X^f$. %Denote %by $E^{f,(i)}$, the matrix $E^f$ generated during round $i$, and by $\vec{D}^{(i)}$, the random matrix $\vec{D}$ generated in round $i$. 
	%  {[$e^f$ and $\vec{E}^{f,(i)}$ were defined already before]} 
    Let $\vec{e}^{f}$ be the matrix of size $\beta\times m\beta$, such that
$$\vec{e}^{f} = \left(\begin{array}{c:c:c}
	\mathbf{0}_{\beta(f-1)} & I_{\beta\times \beta} &\mathbf{0}_{\beta(m-f)}
	\end{array}\right).$$ 
        W.l.o.g., the user can %add the matrix $e^{f}$ to rows $i_n, (i+1)_n, \cdots, (i+\beta)_n$ of $D^{(i)\top}$, where $$i_n = \begin{cases}n \quad\quad\quad\quad\quad~ \text{ if } i= n\\ i \quad (\text{mod } n) \quad \text{ otherwise}\end{cases}.$$
     choose the matrix $E^{f,(i)}$ such that
    \begin{equation}\label{eq:ef}
    \vec{E}^{f,(i)} = \left(\begin{array}{c}
	\mathbf{0}_{i-1\times m\beta}\\
    \vec{e}^{f}\\
    \mathbf{0}_{n-\beta-i+1\times m\beta}
	\end{array}\right),\end{equation}
    and hides it using the random matrix $\vec{D}^{(i)}$ by adding 
    \begin{equation}\label{eq:dq}
    \vec{D}^{f,(i)}_Q = \vec{D}^{(i)\top}+\vec{E}^{f,(i)}\in\Fq^{n\times m\beta},\end{equation}
    where $\vec{D}^{(i)\top}$ is the transpose of matrix $\vec{D}^{(i)}$.
    The user then divides $\vec{D}^{f,(i)}_Q$ into $l$ blocks, where each block, $\vec{D}^{f,(i)}_{Q_j}, j=1,\ldots, l$, consists of $\rho = n/l$ rows of $\vec{D}^{f,(i)}_{Q}$. Afterwards, she lifts every matrix, $\vec{D}^{f,(i)}_{Q_j}$, {(\emph{i.e.}, appends an identity matrix)} to obtain the query matrices
	\begin{equation}
    \vec{Q}_j^{f,(i)} = \left( \begin{array}{c:c}
	\vec{I}_{\rho\times \rho} & \vec{D}_{Q_j}^{f,(i)}
	\end{array}\right) \in \Fq^{\rho\times (\rho+m\beta)}\end{equation}
    for $j=1,\ldots,l$.

%  {[Please don't put matrix sizes as subscript, rather write $\vec{A} \in \Fq^{a \times b}$. For the identity matrix it is common though, so there I think it's ok.]}

{\textbf{Query transmission:}} 
The user sends $\vec{Q}_j^{f,(i)}$ %, \emph{i.e.}, block $j$ of the query matrix $\vec{Q}^{f,(i)}$, consisting of $\rho$ rows, as described in Section~\ref{sec:scheme},
%  {[In Section II C, it was blocks of columns, please be careful!!]}\rt{[Each server stores $\rho$ columns of matrix Y, but the user sends $\rho$ rows of $Q$, should I maybe explain this in more detail?]}\awcomment{Yes, please :-)}
to server $j$, $j = 1,\dots,l$.
The network is random {(cf. Section~\ref{subsec:random-network})} and we assume that server $j$ receives 
%\cami{Be consistent with matrix mult notation, preferably no dot.}
\[\vec{Q}^{f,(i)}_{j,rec} = \left( \begin{array}{c:c}
%The network is random {(cf. Section~\ref{subsec:random-network})} and we assume that server $i$ receives \[\vec{Q}^f_{i,rec} = \left( \begin{array}{c:c}
\vec{A}^{(i)}_j & \vec{A}^{(i)}_j  \vec{D}_{Q_j}^{f,(i)}
\end{array}\right) \in \Fqs^{\rho\times (\rho+m\beta)},\]
where $\vec{A}^{(i)}_j \in \F_{q^s}^{\rho \times \rho}$ is the random {channel} {transfer} matrix in the network from the user to server $j$ {in round $i$}, and $\vec{D}_{Q_j}^{f,(i)}$ is block $j$ %the $i^{th}$ block 
of the matrix $\vec{D}_{Q}^{f,(i)}$, {\emph{i.e.}, the $\rho$ rows of $\vec{D}_Q^{f,(i)}$ sent to server $j$.}

{\textbf{Server response:}} The server projects its data on the query matrix received, and returns to the user the following matrix \[\vec{R}^{f,(i)}_j = \left( \begin{array}{c:c:c}
\vec{I}_{\rho\times \rho} &  \vec{A}^{(i)}_j & \vec{Y}_j\star (\vec{A}^{(i)}_j\vec{D}_{Q_j}^{f,(i)})
\end{array}\right)\in \Fqs^{\rho\times (2\rho+m\beta)},\]
where $\vec{Y}_j$ is block $j$ of $\Y$, as defined in Section~\ref{sec:model}. %\awcomment{[Why is it called $C_i$ and not $Y_i$. Also be more precise, what is "the corresponding block"? What are the sizes of $C_i$?]}\rt{[I guess we can use the contents here, but before, I was thinking to have the codes in the star product, but it is the same, I think.]}

The user then receives the response \begin{align*}
\vec{R}^{f,(i)}_{j,rec} &= \left( \begin{array}{c:c:c}
\vec{A}^{(i)\prime}_j & \vec{A}^{(i)\prime}_j\vec{A}^{(i)}_j & \vec{A}^{(i)\prime}_j (\vec{Y}_j\star (\vec{A}^{(i)}_j\vec{D}_{Q_j}^{f,(i)})
\end{array}\right)\\
&\in \Fqs^{\rho\times (2\rho+m\beta)}\end{align*} from server $j$, after being transmitted through the network {where $\vec{A}^{(i)\prime}_j\in\F_{q^s}^{\rho\times\rho}$ denotes the $\rho\times \rho$ random channel transfer matrix from server $j$ to the user {in round $i$}}. 

{\textbf{Information retrieval:}}
From the responses of the different servers, the user can then form the matrix $\vec{R}^{f,(i)}_{rec}\in \Fqs^{n\times(2n+m\beta)}$.
%\begin{figure*}[tbh]
\begin{align*}\vec{R}^{f,(i)}_{rec}
%&=\small\left(\begin{array}{ccc:ccc:ccc}
%\vec{A}^{(i)\prime}_1 & \cdots & 0 & \vec{A}^{(i)\prime}_1\vec{A}^{(i)}_1 & \cdots & 0 & \vec{A}^{(i)\prime}_1 (\vec{Y}_1\star (\vec{A}^{(i)}_1\vec{D}_{Q_1}^{f,(i)})) & \cdots & 0\\
%\vdots & \ddots & \vdots & \vdots & \ddots & \vdots & \vdots & \ddots & \vdots \\
%%\vdots & \ddots & \vec{A}^{(i)\prime}_{l-1} & 0 & \vdots & \ddots & \vec{A}^{(i)\prime}_{l-1}\vec{A}^{(i)}_{l-1} & 0 & \vec{A}^{(i)\prime}_{l-1} (\vec{Y}_{l-1}\star (\vec{A}^{(i)}_{l-1}\vec{D}_{Q_{l-1}}^{f,(i)})\\
%0 & \cdots & \vec{A}^{(i)\prime}_l & 0 & \cdots & \vec{A}^{(i)\prime}_l\vec{A}^{(i)}_l & 0 & \cdots & \vec{A}^{(i)\prime}_l (\vec{Y}_l\star (\vec{A}^{(i)}_l\vec{D}_{Q_l}^{f,(i)}))\\
%\end{array}\right)\\[2ex]
 &= \left( \begin{array}{c:c:c}
\vec{A}^{(i)\prime} & \vec{A}^{(i)\prime}\vec{A}^{(i)} & \vec{A}^{(i)\prime}(\vec{Y}\star (\vec{A}^{(i)}\vec{D}_{Q}^{f,(i)}))
\end{array}\right)\\
&\in \Fqs^{n\times(2n+m\beta)}\end{align*}
%\end{figure*}
 %\awcomment{[Add the sizes of the newly defined matrices. Also, shouldn't the rightmost matrix also be a diagnoal matrix, otherwise it doesn't work out to write it as matrix product $\vec{A}^{(i)\prime}(\vec{Y}\star (\vec{A}^{(i)}\vec{D}_{Q}^{f,(i)}))$??]}
 
 where \[\vec{A}^{(i)} {:}= \small\left(\begin{array}{cccc}
 A^{(i)}_1 & 0 & \cdots & 0 \\
 0 & \ddots & \ddots & \vdots \\
 \vdots & \ddots & A^{(i)}_{l-1} & 0 \\
 0 & \cdots & 0 & A^{(i)}_{l} \\
 \end{array}\right)\in\F_{q^s}^{n\times n}\]
  and
 $$\vec{A}^{(i)\prime} {:}= \small\left(\begin{array}{cccc}
 A^{(i)\prime}_1 & 0 & \cdots & 0 \\
 0 & \ddots & \ddots & \vdots \\
 \vdots & \ddots & A^{(i)\prime}_{l-1} & 0 \\
 0 & \cdots & 0 & A^{(i)\prime}_{l} \\
 \end{array}\right)\in\F_{q^s}^{n\times n}$$
 
 {\begin{remark}
 Note that \cite{koetter2008coding} deals with single sources only. For multiple sources, in the standard network coding problem, it is not easy to design subspace codes that can be "distributed" to the different sources and combined to one codeword at the receiver side. For example in \cite{Halbawi-multiSource}, 
the problem was solved only for the special case of three sources.
In our scenario, however, the servers that respond to the queries do not encode an \emph{arbitrary} matrix. They transmit the matrix $R_j^{f,(i)}$. Once the user obtains all matrices $R_{j,rec}^{f,(i)}$ these matrices can be combined into the matrix $R_{rec}^{f,(i)}$.
 \end{remark}}
 
\begin{remark} %If $\vec{A}_j$ has full rank then the servers can obtain $\vec{D}_j+\vec{E}^f_j$ and send the lifted matrix $( \begin{array}{c:c}
%\vec{I}_{n\times n} & (\vec{D}_j+\vec{E}_j^f)^T
%\end{array})$ back to the user.
If the overall {channel transfer} matrices $\vec{A}^{(i)}$ %= \left(\begin{array}{cccc}
%A^{(i)}_1 & 0 & \cdots & 0 \\
%0 & \ddots & \ddots & \vdots \\
%\vdots & \ddots & A^{(i)}_{l-1} & 0 \\
%0 & \cdots & 0 & A^{(i)}_{l} \\
%\end{array}\right)\in\F_{q^s}^{n\times n} $$
and 
$\vec{A}^{(i)\prime}$ %= \left(\begin{array}{cccc}
%A^{(i)\prime}_1 & 0 & \cdots & 0 \\
%0 & \ddots & \ddots & \vdots \\
%\vdots & \ddots & A^{(i)\prime}_{l-1} & 0 \\
%0 & \cdots & 0 & A^{(i)\prime}_{l} \\
%\end{array}\right)\in\F_{q^s}^{n\times n}$$
%\awcomment{[You have to define $\vec{A}$]} %formed from the $\vec{A}_j$ matrices 
have full rank, then, the user receives $\vec{A}^{(i)\prime}(\Y\star\vec{D}^{f,(i)}_Q) = \vec{A}^{(i)\prime}\Y\star\vec{A}^{(i)}\vec{D}^{(i)}+\vec{A}^{(i)\prime}\Y\star\vec{A}^{(i)}\vec{E}^{f,(i)}$, which is a codeword {from a} $\Gabcode{n}{k+t\rho-1}$ code with $\beta$ errors in known locations, %\emph{i.e.}, 
which can be treated as erasures. {Since the $\Gabcode{n}{k+t\rho-1}$ code has minimum rank distance $n-k-t\rho+2$, it can correct any $\beta = n-k-t\rho+1$ (rank) erasures} and {the reconstruction of} $\Y\star\vec{E}^{f,(i)}$, \emph{i.e.,} $\beta$ stripes of file $X^f$, can be done exactly as in a usual Gabidulin decoding problem. \end{remark}
%\awcomment{from} a Gabidulin code  \awcomment{of dimension??} \cami{Rather: which is a Gabidulin codeword $\in$ (give code)} as well and everything follows exactly \cami{as in a usual Gabidulin} decoding problem.

%On the other hand, if $A$ is not full rank, then there is a problem since $C*AD$ is not necessarily Gabidulin.

%The user is able to decode $E_{f}$ if the dimension $r = dim(A(\vec{E}^F)) \leq \floor{\frac{(n-k)s}{n}}$. At the same time, $AD$ should have high enough rank.

\begin{example}
	Let  $m$ denote the number of files. The data, $\vec{X} \in\Fqs^{m\times 2}$, is encoded using a $\Gabcode{3}{2}$ Gabidulin code $\mycode{C}$ over $\mathbb{F}_{2^3}$, where 
%	\awcomment{[Can we write this example without generator matrix? Just by evaluation? Before we don't need the generator matrix anymore...]}
	{\[\mycode{C} = \{(f(1),f(\alpha), f(\alpha^2)):f(z) = f_0z+f_1z^2\}.\]}
Hence, its corresponding generator matrix is	
	\[G_{C} = \left( \begin{array}{ccc}
	1 & \alpha & \alpha^2\\
	1 & \alpha^2 & \alpha^2+\alpha
	\end{array}\right)\in \Fqs^{2\times3}.\] 
    %\cami{Boldfacinging matrices looks particularly bad when you don't boldface all of them, like the gen matrix here. Let's not boldface matrices please?} 
    We consider a network with $l=3$ servers, where every column $Y_j$ of $$\vec{Y}=\vec{X} G_C\in\Fqs^{m\times 3}$$ is stored on a server.
	
	The goal is to construct a PIR scheme where the user wants file $X^f$ from $\vec{X}$, while keeping the identity $f$ of the file hidden from the servers, and the servers do not collude, \emph{i.e.}, $t=1$. The number of subdivisions in this case is $\beta=n-k-t\rho+1=1$. Assume the user wants file $X^1$. To this end, we use a \mycode{G}(3,1) code as the query code,
	{\[\mycode{D} = \{(f(1),f(\alpha), f(\alpha^2)):f(z) = f_0z\}\]}
	whose generator matrix is
	\[G_{D} = \left( \begin{array}{ccc}
	1 & \alpha & \alpha^2
	\end{array}\right)\in\Fqs^{1\times 3}.\]
	Two rounds are needed to retrieve the full file $X^1$, hence, let matrix $\vec{e}^{1} = \left(\begin{array}{cccc}
	1 & 0 & \cdots &0\end{array}\right)$. Then, $\vec{E}^{1,(1)}$ is the $3\times m$ matrix     
    $\vec{E}^{1,(1)} = \left(\begin{array}{c}
    \vec{e}^{1}\\
    \mathbf{0}_{n-1\times m}
	\end{array}\right),$ and $\vec{E}^{1,(2)} = \left(\begin{array}{c}
    \mathbf{0}_{1\times m}\\
    \vec{e}^{1}\\
    \mathbf{0}_{n-2\times m}
	\end{array}\right)$
    which the user uses to decode the file $X^1$, as defined in~\eqref{eq:ef}. The matrix $\vec{E}^{1,(i)}$ is added to the random matrix $\vec{D}^{(i)}$, which is chosen uniformly at random from the code $\mycode{D}$, as shown in~\eqref{eq:dq}.
	
	%Since there is only one stripe, in round $\sigma, \sigma=1,2$, $\vec{E}^f$ would be the encoding of the vector $\left( \begin{array}{ccc}
    %1 & 0 & 0
    %\end{array}\right)^\top$ using the Gabidulin code $\mycode{S}=\{g(\alpha_0), g(\alpha_1), g(\alpha_{2}):g(z)=g_0z^{q^{\sigma\beta+\rho t−1}}\}$.

%         \cami{Why is the following de-emphasized? If it's not part of the example, then end the example area here!}

% 	To be more explicit,  assume that the user generates a random vector $\vec{u}\in {\Fq^{m \times 1}}$. The matrix $\vec{D}$ is then \[\vec{D} = \left( \begin{array}{ccc}
% 	u & \alpha u & \alpha^2u
% 	\end{array}\right)\in \cami{\mycode{D}\in} \Fq^{m\times 3}.\] The user adds a vector $e^f_i$ to the $i^{th}$ column of $\vec{D}$, \[\vec{D}+\vec{E}^f = \left( \begin{array}{ccc}
% 	u+e^f_1 & \alpha u+e^f_2 & \alpha^2 u+e^f_3
% 	\end{array}\right).\] For instance, the user can  \cami{send (add?)} the vector of size $m$ containing a single $1$ on position $f$ and zeros elsewhere to the first column of $D$, and add a zero vector to the other columns of $\vec{D}$.
	
	The query matrix is then lifted,
    \[
	\vec{Q}_j^{1,(i)} = \left( \begin{array}{c:c}
	1 & D_{Q_j}^{1,(i)} \\
	%& D_{Q_2}^{(i)} \\
	%& D_{Q_3}^{(i)}
	\end{array}\right),
	\] where $D_{Q_j}^{1,(i)}$ is the $j^{th}$ row of $\vec{D}_Q^{1,(i)}$, generated in round $i,$ for $i=1,2$.
	
	In round $i$, the user sends row $j$ of matrix $\vec{Q}^{f,(i)}$ to server $j$.
	Assume that server $j$ receives $$Q^{f,(i)}_{j, rec} = \left( \begin{array}{c:c}
	a_j^{(i)} & a_j^{(i)}D_{Q_j}^{1,(i)}
	\end{array}\right),$$
	where $a_j^{(i)}\in \F_{2^3}$ is a scalar introduced by the network in round $i$.
	 
	There are two cases a server can encounter:
	\begin{itemize}
		%\vspace{1em}
		%\setlength\itemsep{1.5em}
		\item $a_j^{(i)}=0$: In this case, server $j$ receives the all-zero  vector and sends nothing back to the user.
		\item $a_j^{(i)}\neq 0$: In this case, server $j$ will calculate the star product of its contents %\awcomment{[better: requested data?]}
$\vec{Y}_j$ with the received query matrix, $Q^{f, (i)}_{j,rec}$, and lifts the obtained matrix to form its response 
        
%        \cami{[The transposes don't match here and beneath, the whole star product should be  transposed.]}
%\cami{overflow}
        $$R_j^{1,(i)} = \left( \begin{array}{c:c:c}
		1 & a_j^{(i)} & a_j^{(i)}D_{Q_j}^{1,(i)}
		\end{array}\right),$$ which is sent to the user.
		The user receives \begin{equation*}
        R^{1,(i)}_{j,rec} = \left( \begin{array}{c:c:c}
		{a_j}^{(i)\prime}& {a_j}^{(i)\prime}a_j^{(i)} & {a_j}^{(i)\prime}a_j^{(i)}D_{Q_j}^{1,(i)}
		\end{array}\right)\end{equation*} from  server $j$, $j=1,2,3$, and forms the overall matrix
		$$\vec{R}^{1,(i)} = \small\left(\begin{array}{c:c:c}
		A^{(i)\prime} & A^{(i)\prime}A^{(i)} & A^{(i)\prime}(Y\star A^{(i)}D_{Q_1}^{1,(i)})\\
		\end{array}\right)$$
        where $\vec{A}^{(i)}$ and $\vec{A}^{(i)\prime}$ are the $3\times 3$ diagonal matrices with diagonal elements $(\vec{a}^{(i)}_1, \vec{a}^{(i)}_2, \vec{a}^{(i)}_3)$ and $(\vec{a}^{(i)\prime}_1, \vec{a}^{(i)\prime}_2,\vec{a}^{(i)\prime}_3)$, respectively, as defined in {Section}~\ref{sec:scheme}.\\
%        $\vec{A}^{(i)\prime}= \left(\begin{array}{ccc}
%    \vec{a}^{(i)\prime}_1 & 0 & 0\\
%    0 & \vec{a}^{(i)\prime}_2 & 0\\
%    0 & 0 & \vec{a}^{(i)\prime}_3\\
%    \end{array}\right)$ and 
%    $\vec{A}^{(i)} = \left(\begin{array}{ccc}
%    \vec{a}^{(i)}_1 & 0 & 0\\
%   0 & \vec{a}^{(i)}_2 & 0\\
%    0 & 0 & \vec{a}^{(i)}_3\\
%    \end{array}\right)$, as defined in {Section}~\ref{sec:scheme}.\\
        
        The user may now face two different scenarios: 
	 	\begin{itemize}
	 		\item All coefficients ${a_j}^{(i)\prime}\neq 0$ for all $j\in\{1,2,3\}$. In this case, the user can decode the matrix $$%\left( \begin{array}{c}
	 		%\vec{Y}_1\star (D_1^{(i)}+e^{f,(i)}_1)\\
	 		%\vec{Y}_2\star (D_2^{(i)}+e^{f,(i)}_2)\\
	 		%\vec{Y}_3\star (D_3^{(i)}+e^{f,(i)}_3)
	 		%\end{array}\right) = 
            \left(Y\star D_{Q_1}^{1,(i)}\right)=
            \left( \begin{array}{c}
	 		\vec{Y}\star \vec{D}^{(i)\top}+\vec{Y}\star \vec{E}^{1,(i)}
	 		\end{array}\right).$$ The matrix $\vec{Y}\star \vec{D}^{(i)\top}$ is a codeword from a $\mycode{G}(3,2)$ code $\mycode{C}\star\mycode{D}$, since $\Y$ is a codeword from a $\mycode{G}(3,2)$ %\awcomment{$\mycode{G}(3,2)$ ??} 
code and $\vec{D}^{(i)}$ is from a $\mycode{G}(3,1)$ code. Therefore, the user can decode one rank erasure %\awcomment{[I removed "or deviation" here as we haven't used it before and haven't explained it.]}
            in $\mycode{C}\star \mycode{D}$. %error if its position is known
%	 		That is, the user can decode if only one of the vectors $e^{f,(i)}_j$ is  nonzero.  
In our setting, the matrix $E^{1,(i)}_j$ is the deterministic matrix introduced by the user in order to decode the desired file. The positions of the errors are known, so they can be viewed as erasures, thus allowing the user to decode $\vec{Y}\star \vec{E}^{1,(i)}$.
	 		\item On the other hand, if at least one of the coefficients ${a_j}^{(i)\prime}{=0}$, then $A^{(i)\prime}(\Y\star\vec{D}^{(i)\top}) $ becomes a codeword in a $\mycode{G}(2,2)$ Gabidulin code that is not able to correct any errors {or erasures}. Thus, the user will not be able to decode the desired file. This case is equivalent to the case where ${a_j}^{(i)}=0$.
	 		For instance, assuming ${a_2'}^{(i)}=0$, the user obtains $\vec{R}^{1,(i)} = \small\left(\begin{array}{c:c:c}
	 		A^{(i)\prime} & A^{(i)\prime}A^{(i)} & A^{(i)\prime}(Y\star A^{(i)}D_{Q_1}^{1,(i)})\\
	 		\end{array}\right)$ %\small\left(\begin{array}{ccc:ccc:ccc}
%		{a_1}^{(i\prime)}&0&0 & {a_1}^{(i)\prime}{a_1}^{(i)}&0&0 & {a_1}^{(i)\prime}a_1^{(i)}D_{Q_1}^{1,(i)}&0&0\\
%		0&0&0 & 0 & 0&0&0&0&0\\
%		0&0&{a_3}^{(i)\prime} & 0&0&{a_3}^{(i)\prime}{a_3}^{(i)} & 0&0&{a_3}^{(i)\prime}a_3^{(i)}D_{Q_3}^{1,(i)}
%		\end{array}\right)$$
		such that the second row is the all-zero vector,
%		\begin{align*}&A^{(i)\prime}(Y\star A^{(i)}D_{Q_1}^{1,(i)})\\&=\left(\begin{array}{ccc}
%		{a_1}^{(i)\prime}a_1^{(i)}D_{Q_1}^{1,(i)}&0&0\\
%		0 & 0 & 0\\
%		0&0&{a_3}^{(i)\prime}a_3^{(i)}D_{Q_3}^{1,(i)}\\
%		\end{array}\right)\end{align*}
and one can observe that the user cannot decode the file from this response matrix.          
	 	\end{itemize}
    \end{itemize}
	 
	In this example, the PIR rate is $\frac{1}{3}$  {with error probability $1-(1-\frac{1}{8})^{12}=1-\left(\frac{7}{8}\right)^{12}\approx 0.798582762$}. %Otherwise, the PIR rate is zero. \cami{fix this!} {With this network realization, on average, the ratio of the downloaded information symbols to the total number of downloaded symbols is  $((1-\frac{1}{q^s})^6)^2=(1-\frac{1}{8})^{12}=\frac{1}{3}(1-\frac{1}{q^s})^{12}=\frac{1}{3}(\frac{7}{8})^{12}\approx 0.067.$}
	%The probability that none of the $a_j{^{(i)}}$'s and the $a_j{^{(i)}}'$'s are zero in both rounds, and thus, that decoding is successful, is $((1-\frac{1}{q^s})^6)^2=(1-\frac{1}{8})^{12}$. Therefore, the PIR rate of this scheme is $\frac{1}{3}(1-\frac{1}{q^s})^{12}=\frac{1}{3}(\frac{7}{8})^{12}\approx 0.067.$ 
	%\cami{fix this!}
	Nevertheless,  {the error probability} approaches $0$ %the asymptotically optimal PIR rate $\frac{1}{3}$ 
	when $q$ is sufficiently large  {(cf. Prop.~1)}.
	%\cami{I would just say $q$ sufficiently large.}
%\cami{[As we are assuming a random network, it would be nice to briefly comment on the probability of the above cases. It is not really fair to say that "when decoding is successful, the rate is..", since when it is not successful it will yield zero rate. For instance, if I do one transmission on rate 1, and one on rate 0, then it's fairer to say that the rate is 1/2 (on average). In general, we should rather define the rate as an expectation of a random variable, which depends on the realization of the $a_i$'s and ($a_i'$)'s and the related probabilities.]}
\end{example}

In the following example, we assume that the servers store $2$ columns of the data $\Y$, and allow $2$ servers to collude. Moreover, it is shown that if erasures occur in the network, partial recovery of the user's requested file is still possible.

%\awcomment{[NEW COMMENT: The previous sentence is fine, but what can the reader learn from Ex. 3 that he couldn't learn from Ex. 2? That partial recovery is possible?]}

\begin{example}\label{ex:eras}
%\cami{[Rewrite in proper passive tense or at least avoid excessive use of `we' passive. Also, this is not a full 4-collusion, as it consists of (partially) known (disjoint) patterns. Should we address this in the light of our ISIT'17 paper?]}
	Assume that the data $\X\in \Fqs^{m\times 3}$ is encoded using a $\mycode{G}(8,3)$ Gabidulin code $\mycode{C}$, with generator matrix $G_C$ over $\F_{2^8}$. The data is then stored on $l=4$ servers, such that each server stores $\rho=2$ columns of $\Y=\X G_C$. Let $\Y_j$ be the block of $\Y$ stored on server $j$.
	
	Assume also that any $t=2$ servers can collude. Since every server is storing two columns of the data, as discussed in~\ref{sec:model}, this can be seen as a case of $4$-collusion. %\rt{Technically, the collusion is not full $4$-collusion, as not any $4$ \textcolor{red}{sub-servers [should define this]} can collude, in \cite{tajeddine2017private}, PIR schemes with colluding patters is considered. For this work, however, we will assume full $4$-collusion.} 
    The goal is to construct a PIR scheme with $4$-collusion that allows a user to retrieve a file $X^f$ from the servers, without revealing the file identity to any of the servers. The files are subdivided into $\beta=n-k-t\rho+1=2$ stripes, and $k=3$ rounds of queries are sent. In round $i$, the query code %\cami{if we specify the field for  C, should also specify it for D. Maybe don't specify at all? That also makes the rate look nicer in the end.}
    \[\mycode{D} = \{d = (g(\alpha_0),g(\alpha_1)):g(z) = g_0 z + g_1 z^2\}\in \F_{2^8}^{\mu\times 8} \] is used to encode $t\rho=4$ random vectors $u_1^{(i)\prime}, u_2^{(i)\prime}, u_3^{(i)\prime}, u_4^{(i)\prime}\in\F_{2}^{2m}$ and form the random matrix $\vec{D}^{(i)}\in\F_2^{m\beta\times n}$.
	
%	\left( \begin{array}{cccccccc}
%	1 & \alpha & \alpha^2 & \alpha^3 & \alpha^4 & \alpha^5 & \alpha+1 & \alpha^2+1 \\
%	1 & \alpha^2 & \alpha^4 & \alpha^5+1 & \alpha^5+\alpha^2+\alpha+1 & \alpha^5+\alpha^4+\alpha^3+\alpha^2+\alpha + 1 & 
%	\alpha^2+1 & \alpha^4+1 \\
%	1 & \alpha^4 & \alpha^5+\alpha^2+\alpha+1 & \alpha^5+\alpha^4+\alpha^3+\alpha^2+\alpha & \alpha^5+\alpha^3+\alpha & \alpha^5+\alpha^3+\alpha^2+1 & 
%	\alpha^4+1 & \alpha^5+\alpha^2+\alpha+1\\
%	1 & \alpha^5+\alpha^2+\alpha+1 & \alpha^5+\alpha^2+\alpha+1 & \alpha^5+\alpha^4+\alpha^3+\alpha^2+\alpha & \alpha^5+\alpha^3+\alpha & \alpha^5+\alpha^3+\alpha^2+1 & 
%	\alpha^4+1 & \alpha^5+\alpha^2+\alpha+1\\
%	\end{array}\right).\]
	
	Afterwards, the $3\times 2m$ deterministic matrix $\vec{E}^{f,(i)}$, chosen as explained in Section~\ref{sec:scheme}, is added to the random matrix $\vec{D}^{(i)}$ to form $\vec{D}_Q^{f,(i)}=D^{(i)\top}+E^{f,(i)}$. 
	
	The query matrix to server $j$ is then \[
	\vec{Q}^f_j = \left( \begin{array}{c:c}
	I_{2\times 2} & \vec{D}_{Q_j}^{f,(i)} \\
	%& \vec{D}_{Q_2}^{f,(i)} \\
	%& \vec{D}_{Q_3}^{f,(i)} \\
	%& \vec{D}_{Q_4}^{f,(i)}
	\end{array}\right),\] where $\vec{D}_{Q_j}^{f,(i)}$ is the $j^{th}$ block, consisting of $2$ rows of $\vec{D}_{Q}^{f,(i)}$.
	
    For instance, assume that the user wants file $X^1$. Then in the first round, %\[\vec{E}^{1,(1)}=\left(\begin{array}{cc:ccc}
    %1&0&0&\cdots&0\\
    %0&1&0&\cdots&0\\
    %0&0&0&\cdots&0\\
    %0&0&0&\cdots&0\\
    %0&0&0&\cdots&0\\
    %0&0&0&\cdots&0\\
    %0&0&0&\cdots&0\\
    %0&0&0&\cdots&0\\
    %\end{array}\right).\]
    \[\vec{E}^{1,(1)}=\left(\begin{array}{c}
    e^1\\
    \mathbf{0}_{6\times 2m}
    \end{array}\right).\]
    
	The user sends $Q_j^{1,(i)}$, \emph{i.e.,} block $j$ of the query matrix $Q^{1,(i)}$, consisting of $j$ rows, to server $j$. Server $j$ receives the matrix $$Q_{j,rec}^{1,(i)} = \left( \begin{array}{c:c}
	\vec{A}_j^{(i)} & \vec{A}_j^{(i)}\vec{D}_{Q_j}^{1,(i)}
	\end{array}\right),$$
	where $\vec{A}_j^{(i)}\in\Fqs^{2\times 2}$ is the $2\times 2$ matrix introduced by the network in round $i$. %\awcomment{[time instance? I don't mind, both is fine for me, but it should be consistent]} $i$.
	Then, the server sends back the matrix $$\vec{R}^{1,(i)}_j = \left( \begin{array}{c:c:c}
	I_{2\times 2} & \vec{A}_j^{(i)} & \vec{Y}_i\star(\vec{A}_j^{(i)}\vec{D}_{Q_j}^{1,(i)})
	\end{array}\right),$$
and the user receives the matrix $$\vec{R}^{1,(i)}_{j,rec} = 
	\left( \begin{array}{c:c:c}
\vec{A}^{(i)\prime}_j & \vec{A}^{(i)\prime}_j\vec{A}^{(i)}_j & \vec{A}^{(i)\prime}_j (\vec{Y}_j\star (\vec{A}^{(i)}_j\vec{D}_{Q_j}^{1,(i)})
\end{array}\right)$$ where $\vec{A}^{(i)\prime}_j$ is the $2\times 2$ matrix introduced by the network from server $j$ to the user. Then the user can construct the matrix \begin{align*}\vec{R}_{rec} &= \left( \begin{array}{c:c:c}
\vec{A}^{(i)\prime} &\vec{A}^{(i)\prime}\vec{A}^{(i)} & \vec{A}^{(i)\prime}(Y\star(\vec{A}^{(i)}D_Q^{1,(i)}))
% {cccc:cccc:c}
% 	\vec{A}^{(i)\prime}_1 & 0 & 0 & 0 & \vec{A}^{(i)\prime}_1\vec{A}^{(i)}_1 & 0 & 0 & 0 & \vec{A}^{(i)\prime}_1 (\vec{Y}_1\star (\vec{A}^{(i)}_1\vec{D}_{Q_1}^{f,(i)})\\
% 0&\vec{A}^{(i)\prime}_2 & 0 & 0 & 
% 0& \vec{A}^{(i)\prime}_2\vec{A}^{(i)}_2 & 0 & 0 & \vec{A}^{(i)\prime}_2 (\vec{Y}_2\star (\vec{A}^{(i)}_2\vec{D}_{Q_2}^{f,(i)})\\
% 0&0&\vec{A}^{(i)\prime}_3 & 0 & 
% 0&0& \vec{A}^{(i)\prime}_3\vec{A}^{(i)}_3 & 0 & \vec{A}^{(i)\prime}_3 (\vec{Y}_3\star (\vec{A}^{(i)}_3\vec{D}_{Q_3}^{f,(i)})\\
% 0&0&0&\vec{A}^{(i)\prime}_4 & 
% 0&0&0& \vec{A}^{(i)\prime}_4\vec{A}^{(i)}_4& \vec{A}^{(i)\prime}_4 (\vec{Y}_4\star (\vec{A}^{(i)}_4\vec{D}_{Q_4}^{f,(i)})\\
	\end{array}\right),
    \end{align*}
	where $\vec{A}^{(i)\prime} = \small\left(\begin{array}{cccc}
    \vec{A}^{(i)\prime}_1 & 0 & 0 & 0\\
    0 & \vec{A}^{(i)\prime}_2 & 0 & 0\\
    0 & 0 & \vec{A}^{(i)\prime}_3 & 0\\
    0 & 0 & 0 & \vec{A}^{(i)\prime}_4\\
    \end{array}\right)$ and $\vec{A}^{(i)} = \small\left(\begin{array}{cccc}
    \vec{A}^{(i)}_1 & 0 & 0 & 0\\
    0 & \vec{A}^{(i)}_2 & 0 & 0\\
    0 & 0 & \vec{A}^{(i)}_3 & 0\\
    0 & 0 & 0 & \vec{A}^{(i)}_4\\
    \end{array}\right)$, as defined in {Section}~\ref{sec:scheme}.
	
	We see that $\mycode{C}\star \mycode{D}$ is a $\mycode{G}(8,6)$ Gabidulin code, which can correct up to 2 erasures. The matrix $\vec{E}^{f,(i)}$ adds $2$ errors with known locations to the code, {\emph{i.e.,} erasures}. Thus, the user is able to decode the file $X^f$ if both matrices $\vec{A}^{(i)\prime}_j$ and $\vec{A}^{(i)}_j$ have full rank in all rounds with a PIR rate $R_{\textup{\textsf{PIR}}} = \frac{\rk_{q^s}(A^{(i)\prime}\Y\star A^{(i)}\vec{E}^{f,(i)})}{n}=\awcomment{\frac{1}{4}}$. Following Proposition~\ref{prop:prob}, this happens with probability $\left((1-\frac{1}{q^s})^{16}\right)^3=(1-\frac{1}{q^s})^{19}$. %Thus, the  {the average number of bits retrieved} is %for the user to retrieve the full file $X^f$ is
	%$\frac{2}{8}(1-\frac{1}{q^s})^{19}\approx 0.23$.
    
    On the other hand, assume that the matrices $A^{(i)}$ and $A^{(i)\prime}$ have a rank deficiency, such that $\rk_{q^s}(\vec{A^{(i)\prime}(Y}\star \vec{A^{(i)}D_Q^{f,(i)}}))=7$. The user can then still decode $1$ stripe of the required file if the erasures from the network occur on the links where the errors from the matrix $E^{f,(i)}$ are added, \emph{i.e.,} the rank deficiency occurs in the matrix $A^{(i)\prime}(\Y\star A^{(i)}\vec{E}^{f,(i)})$. %\cami{use "backshalsh left" and right to make the parenthesis below correct size} 
    %, and thus, $\rk_{q^s}(\vec{A}^{(i)\prime}_j(\Y\star A^{(i)}\vec{E}^{f,(i)}))=1$. 
   
   %\cami{ANTONIA: can you see if you agree w the probability below? What would be a reasonable assumption for $q_A$, field size for $A$?} \awcomment{The probability is fine. I don't know which field size is reasonable, in network coding literature examples, many different field sizes are used...}
   
    Let $P_1$ denote the probability of this event.\footnote{This is the sum of two terms: The first is the probability that both $A^{(i)}$ and $A^{(i)\prime}$ have an erasure in the same place and in one of the two places where $E^{f,(i)}$ is nonzero, and the second is the probability that one of the matrices has an erasure in a position where $E^{f,(i)}$ is nonzero and the other is full rank. }
     \begin{align*}
    P_1&=\mathbb{P}\left(\rk_{q^s}\left(\vec{A^{(i)\prime}(Y}\star \vec{A^{(i)}D_Q^{f,(i)}})\right)=7 \right.\\&\left.~~~~~~~\textit{ and } \rk_{q^s}\left(A^{(i)\prime}(\Y\star A^{(i)}\vec{E}^{f,(i)})\right)=1\right)\\
    &=\binom{2}{1}\left(\left(1-\frac{1}{q^s}\right)^7\left(\frac{1}{q^s}\right)\right)^2\\&~~~~~~~+2\left[\binom{2}{1}\left(\left(1-\frac{1}{q^s}\right)^7\left(\frac{1}{q^s}\right)\right)\left(1-\frac{1}{q^s}\right)^8\right]\\
    %&=2\left(\left(1-\frac{1}{q^s}\right)^7\left(\frac{1}{q^s}\right)\right)\left(\left(1-\frac{1}{q^s}\right)^7\left(\frac{1}{q^s}\right)+2\left(1-\frac{1}{q^s}\right)^8\right)\\%\stackrel{q^s=2^8}
    &\approx 0.015.\end{align*} %The probability $P_1$ is the probability that the erasures from the network occur on the links where the errors from the matrix $E^{f,(i)}$ are added, \emph{i.e.,} the rank deficiency occurs in the matrix $A'(\Y\star A\vec{E}^{f,(i)})$. 
    Then, the user is able to decode $1$ stripe, %\awcomment{[sth. missing here? What is the user able to decode? One stripe?]}, 
    %\awcomment{=[give number]}$ stripe when $\rk_{q^s}(\vec{A'(Y}\star \vec{A(D_Q^{f,(i)})}) - \rk_{q^s}(A'\Y\star A\vec{E}^{f,(i)})\geq k$, and $\rk_{q^s}(A'\Y\star A\vec{E}^{f,(i)})\geq 1$, 
    with PIR rate $R_{\textup{\textsf{PIR}}} 
    %= P_1\cdot\frac{\rk_{q^s}(A'(\Y\star A\vec{E}^{f,(i)}))}{8}
    =\frac{1}{8}$, with error probability $1-P_1$.
    Thus, the  {the average realized download rate is} %of this scheme can be written as 
	\[P_2\cdot\frac{2}{8}+P_1\cdot\frac{1}{8}\approx 0.24\]
    where $P_2$ is the probability that both $A^{(i)}$ and $A^{(i)\prime}$ have full rank in round $i$, \emph{i.e.,} $P_2=\left(1-\frac{1}{q^s}\right)^{16}\approx 0.94$.
\end{example}

%\cami{We can use this prop. in the thms below.}

%\cami{bibtex entry:\cite{Ho_fullrank}
% @ARTICLE{Ho_fullrank,
% author={T. Ho and M. Medard and R. Koetter and D. R. Karger and M. Effros and J. Shi and B. Leong},
% journal={IEEE Transactions on Information Theory},
% title={A Random Linear Network Coding Approach to Multicast},
% year={2006},
% volume={52},
% number={10},
% pages={4413-4430},
% month={Oct},}
%}

% \cami{[Old version.]}
% \begin{theorem}
% The user can decode the file from the responses if $\rank(\vec{A}'(\vec{Y}\star \vec{A}(\vec{D}+\vec{E}^f)^T)) - \rank(\vec{E}^f)\geq k$.
% \end{theorem}
% \cami{[Much better!]}

%\cami{simplify/clarify? What are the assumptions?}
According to Prop. \ref{prop:prob}, the probability of matrices $\vec{A}^{(i)}$ and $\vec{A}^{(i)\prime}$ are simultaneously full rank in all query rounds $i=1,\ldots, k$ is $\left(1-\frac{1}{q^s}\right)^{2n+k}$. Together with the results of the previous examples, we can conclude the following theorems. Since the scheme from one round is the same for all other rounds, we will drop the superscript $(i)$ in the following.
\begin{theorem}\label{th:decode}
%Assume that $\rank\awcomment{_{q^s}}(\vec{A}'(\vec{Y}\star \vec{A}\vec{D}_Q^{f})) - \rank\awcomment{_{q^s}}(\vec{Y}\star\vec{E}^f)\geq k+t\rho-1$. 
In the above setting, the user can decode the file $X^f$ from the $k$ responses from rounds $i=1,\ldots,k$, $R^{f}_{rec}$, with PIR rate 
%$1-\frac{k+t\rho-1}{n}$. 
%\cami{fix to $q^s$ everywhere!}
$$R_{\textup{\textsf{PIR}}} =\left(1-\frac{k+t\rho-1}{n}\right),$$  {with error probability $1- \left(1-\frac{1}{q^s}\right)^{2n+k}$.}

%Taking the probability of successful decoding into account, that is, the probability that both channel transfer matrices $\vec{A}$ and $\vec{A}'$ are full rank, the PIR rate becomes $\left(1-\frac{1}{q^s}\right)^{2n+k}\left(1-\frac{k+t\rho-1}{n}\right)$.
\end{theorem}

\begin{proof} \emph{Decodability.} Assume that $\rank_{q^s}(\vec{A}) = \rank_{q^s}(\vec{A'})=n$, \emph{i.e.}, $\vec{A}$ and $\vec{A'}$ are full rank, which happens in all $k$ rounds with probability $\left(1-\frac{1}{q^s}\right)^{2n+k}$, as given in Prop.~\ref{prop:prob}. The responses $R^{f}_{rec}$ could then be reduced to $$\vec{Y}\star\vec{D}+\Y\star \vec{E}^f,$$ where $\Y\star \vec{D}_Q^f$ is a codeword from \mycode{G}($n,k+t\rho-1$) Gabidulin code, and $\rank_{q^s}(\Y\star \vec{E}^f) = n-k-t\rho+1$. Therefore, the user can decode $\Y\star \vec{E}^f$, which gives one part from each stripe in each round (this follows from the way the matrix $\vec{E}^f$ is chosen). Moreover, the choice of $\vec{E}^f$ ensures that the user retrieves at most one part of any stripe from a single sub-server, \emph{i.e.,} the user does not retrieve redundant information about a certain stripe. Hence, after $k$ query rounds, the user will have $k$ independent equations for each stripe, and therefore, can recover the full file $X^f$.

\emph{Privacy.} Since the random matrix $\vec{D}$ is a \mycode{G}($n,t\rho$) Gabidulin code, from which the codewords are chosen uniformly at random, the query received by any $t\rho$ servers is random, and the server cannot know where a $1$ was added.
\end{proof}

\begin{remark}
It is easy to see that  {the error probability in the above theorem approaches $0$ when the field size grows.}
Moreover, the PIR rate  {matches} the conjectured capacity \cite{freij2016private}.
%\[R_{\textup{\textsf{PIR}}}\underset{q\to \infty}{\longmapsto}1-\frac{k+t\rho-1}{n}.\]
\end{remark}

%\cami{include probabilities, so it's not $1/8$.}
The above theorem is restricted to the case where no erasures occur in the network. However, as was shown in Example~\ref{ex:eras}, %when $\rank_{q^s}(\vec{A}'(\vec{Y}\star \vec{A(D_Q^f)})=7$ and $\rank_{q^s}(\vec{A}'(\vec{Y}\star \vec{A}\vec{E}^{f}))=1$, the user can still decode one stripe of the requested file. In general,
the rank deficiency of $\vec{A}'(\vec{Y}\star \vec{A}\vec{D}_Q^{f})$ might reduce the rank of $\vec{A}'(\vec{Y}\star \vec{A}\vec{E}^{f})$ that is received by the user as well. In such cases, the user can still decode $\rank_{q^s}(\vec{A}'(\vec{Y}\star \vec{A}\vec{E}^{f}))$ stripes from the requested file, $X^f$. For instance, as shown in Example~\ref{ex:eras}, %when $\rank_{q^s}(\vec{A}'(\vec{Y}\star \vec{A(D_Q^f)})=7$ and $\rank_{q^s}(\vec{A}'(\vec{Y}\star \vec{A}\vec{E}^{f}))=1$, %despite the rank deficiency of the matrix $\vec{A}'(\vec{Y}\star \vec{A(D_Q^f)}$, 
the user can still be able to decode one stripe of the requested file, despite the rank deficiency of the matrix $\vec{A}'(\vec{Y}\star \vec{AD_Q^f)}$. %\rt{The PIR rate was then $R_{\textup{\textsf{PIR}}} =P_2\cdot\frac{2}{8}+P_1\cdot\frac{1}{8}$ \cami{$=$????? Again, choose either the general expression or the specific example which lead to a concrete number here.} as described in Example~\ref{ex:eras}.}

{In} the following, we consider the rate given from a single query round $i$. Since the probabilistic rate is the same for all query rounds, $i=1,\cdots, k$, we discuss what follows for a single round and %the PIR rate of the scheme will be the same as the PIR rate for one round, thus,
drop the superscript $(i)$. In the first stage of querying, if no erasures occur, then the user can still retrieve $\beta$ parts of the requested file $X^f$. On the other hand, assume that $\rk_{q^s}(\vec{A}'(\vec{Y}\star \vec{AD_Q^f})) - \rk_{q^s}(A'(\Y\star A\vec{E}^f))= k+t\rho-1,$ and $\rk_{q^s}(A'(\Y\star A\vec{E}^f))=\delta<\beta$, the user can then retrieve $\delta$ parts of the file. However, the user has to send a second stage of queries to retrieve $(\beta-\delta)$ new parts of the file $X^f$, which can substitute the erased parts. To do that, the user will subdivide the stripes again into $\beta$ sub-stripes, then use the same scheme as in the first stage to retrieve a part of each sub-stripe. Here it is important to note that each sub-server should be asked at most once about any stripe. This means that, in the second stage of querying, the matrix $E^f$ should be chosen such that, to retrieve stripe $\ell$, a $1$ is never added to the columns $\ell, \ldots, \ell+k-1$ of $\vec{D}$, since these will give information already acquired about stripe $\ell$ in other rounds, \emph{i.e.,} they give redundant information. Assume that no erasures occur in the second stage, then, along with the responses from the first stage, the user now has $k$ independent equations about all sub-stripes of the file $X^f$. On the other hand, if erasures occur in stage 2, this process will be followed another time. This will be done recursively until all the file $X^f$ is recovered. 
%For instance, let us assume the response from sub-server 1 was erased in round $1$. By construction this means that the user is not able to recover anything about stripe $1$ in this round. The user then subdivides the storage into $\beta$ sub-stripes
We state this more formally in Theorem~\ref{th:eras} below.

%\begin{figure*}[tbh]
Denote
\begin{align*}
\mathbb{P}&(\rk_{q^s}(\vec{A'(Y}\star \vec{A D_Q^{f}})) - \rk_{q^s}\left(A'(\Y\star A\vec{E}^f)\right)= k+t\rho-1,\\&~\text{and }\rk_{q^s}\left(A'(\Y\star A\vec{E}^f)\right)=\delta )=: P_\delta.
\end{align*}
%\end{figure*}

\begin{theorem}\label{th:eras}
The user can decode $\delta$ stripes of the file $X^f$ from the responses $\vec{R}^{f}_{rec}$ with a PIR rate  
\begin{equation}\label{eq:pirprob}
R_{\textup{\textsf{PIR}}}=\frac{\delta}{n}\,,
\end{equation}
 {with error probability $1-P_\delta$.}

%For the user to retrieve the full file, the PIR rate is shown in eq.~\eqref{eq:rateprob}. 

%\begin{figure*}
%\begin{equation}\label{eq:rateprob}
%R_{\textup{\textsf{PIR}}}=\sum_{\delta_1=1}^{\beta}\frac{\frac{\delta_1}{\beta}\cdot P_{\delta_1}\cdot \delta_1 + \frac{\beta-\delta_1}{\beta}\sum_{\delta_2=1}^{\beta}\left(\frac{\delta_2}{\beta}\cdot P_{\delta_2}\cdot\frac{\delta_2}{\beta}+ \frac{\beta-\delta_2}{\beta}\sum_{\delta_3=1}^{\beta}\left(\frac{\delta_3}{\beta}\cdot P_{\delta_3}\cdot\frac{\delta_3}{\beta^2}+\cdots\right)\right)}{\sum_{i=0}^\infty \frac{n}{\beta^i}}\,.
%\end{equation}
%\end{figure*}

\end{theorem} 

\begin{proof}
\emph{Decodability.} Since $Y\star \vec{D^{f}}$ is a codeword in a $\mycode{G}(n,k+t\rho+1)$ Gabidulin code, it is able to tolerate $n-(k+t\rho-1)$ erasures. Therefore, the user can decode the codeword $Y\star \vec{D^{f}}$, whenever $\rk_{q^s}(\vec{A'(Y}\star \vec{A D_Q^{f}})) - \rk_{q^s}(A'(\Y\star A\vec{E}^f))= k+t\rho-1$. Additionally, assume that $\rk_{q^s}(A'(\Y\star A\vec{E}^f))=\delta$. Then the user is able to retrieve $\delta$ stripes of the requested file $X^f$. 
As the probability of the joint event that $\rk_{q^s}(\vec{A'(Y}\star \vec{A D_Q^{f}})) - \rk_{q^s}(A'(\Y\star A\vec{E}^f))= k+t\rho-1$ and $\rk_{q^s}(A'(\Y\star A\vec{E}^f))=\delta$ is $P_\delta$, the PIR rate is 
\begin{equation}\label{eq:pirprob}
R_{\textup{\textsf{PIR}}}=\frac{\delta}{n}\,,
\end{equation}
with error probability $1-P_\delta$.
 {The average realized download rate is
$$\sum_{\delta=1}^{\beta}\frac{P_\delta\cdot\delta}{n}.$$}

%To retrieve the stripes erased in the first stage, the user then subdivides the stripes into $\beta$ sub-stripes and sends another $k$ queries in the same manner as in the first stage. Similarly to the first stage, the PIR rate for the second stage is $$
%R_{\textup{\textsf{PIR}}}=\sum_{\delta_2=1}^{\beta}\frac{P_{\delta_2}\cdot\frac{\delta_2}{\beta}}{\frac{n}{\beta}}.$$ Combining the rates of the first two stages together, the PIR rate becomes \[R_{\textup{\textsf{PIR}}}=\sum_{\delta_1=1}^{\beta}\frac{\frac{\delta_1}{\beta}\cdot P_{\delta_1}\cdot \delta_1+ \frac{\beta-\delta_1}{\beta}\sum_{\delta_2=1}^{\beta} P_{\delta_2}\cdot\frac{\delta_2}{\beta}}{n+\frac{n}{\beta}}. \] Continuing with this process recursively gives the PIR rate given in~\eqref{eq:rateprob}.

\emph{Privacy.} Privacy is achieved for the same reason as in Theorem~\ref{th:decode}.

\end{proof}

\section{PIR on a network with errors}
%\awcomment{[I didn't go thorugh this section in detail, but I will do it once all the current comments are addressed :-)]}
%As our model is for a random network, some errors in the network might occur. 
Errors in a random network happen, \emph{e.g.,} due to malicious nodes injecting erroneous packets or due to congestion.
In this section, we will consider such an erroneous network. The network is assumed to introduce up to $\epsilon$ errors and up to $\tau$ erasures to the sent packets.
%\cami{\sout{have to take these network errors $\vec{N}_e$ into account.} [Define the error model properly.]} 

From \cite{silva2008rank}, it is known that errors of rank $\epsilon$ can be corrected if $2 \epsilon\leq \tau+d-1$, %where $\tau$ is the number of erasures and $\gamma$ is the number of deviations occurring in the network, and 
where $d$ is the minimum rank distance of the code. %In this system, we assume $\epsilon$ errors and $\tau$ erasures occur in the network. 
For this purpose, each file is subdivided into $\beta=n-k-\rho t-2\epsilon-\tau+1$ stripes. %, and the added matrix $\vec{E}^f$ introduces errors with known locations, \emph{i.e.}, those errors can be viewed as erasures.
The matrix $D^{(i)}$ is constructed in the same manner as in section~\ref{sec:scheme}. As for the matrix $\vec{E}^{f,(i)}$, it is chosen as a codeword of a rank metric code $\mycode{E}^f$ such that the response $\mycode{C}\star(\mycode{D}+\mycode{E}^f)$ is again a Gabidulin code, which follows from the techniques used in \cite{tajeddine2018private}. 
Specifically, we start with $\vec{E}^f$ being the $n\times m\beta$ all-zero matrix. In round $i$, for all $\delta\leq \ceil{\frac{i\beta}{k}}$, the vector $e^{f,(i)}_{\delta}$ is taken to be the vector with all zeros and a single $1$ at position $\beta(f-1)+\delta$, which is the vector requesting the stripe $\delta$ of the file $X^f$. The vector $e^{f,(i)}_{\delta}$ is encoded using a $\mycode{G}(n,1)$ code $\mycode{E}^f=\{g(\alpha_0), \cdots, g(\alpha_{n-1}):g(z)=g_0z^{q^{i\beta-\delta k+k+\rho t-1}}\}$ and added to row $(f-1)m\beta+\delta$ of $E^{f,(i)}$.\footnote{For more specific details on how the scheme is constructed, see \cite{tajeddine2018private}.}

\begin{remark}
	{In the above scheme, the storage code $\mycode{C}$ is a Gabidulin code of dimension $k$. Hence, we pick $\mycode{D}$ to be a Gabidulin code of dimension $t\rho$, so that $\mycode{C}\star\mycode{D}$ is a Gabidulin code of dimension $t\rho+ k-1$ (see Section~\ref{subsec:star-prod}).		
		%Then we pick \awcomment{the matrix} ${E}^f$ from a Gabidulin code \awcomment{$\mycode{G}(?,?)$ over $??$} such that $\mycode{C}\star\mycode{D}$ and $\mycode{C}\star{E}^f$ intersect only trivially. 
		Thus, in round $i$, $\mycode{C}\star(\mycode{D}+\mycode{E}^f)$ is a $\mycode{G}(n,k+t\rho+\beta-1)$ Gabidulin code, padded with parts retrieved from previous rounds $1,\ldots, i - 1$, with degrees $>n-1$, where degrees $q^{t\rho+k}$ and higher of the evaluation polynomial consist of parts of the file requested by the user. From $\mycode{C}\star(\mycode{D}+\mycode{E}^f)$, the user can interpolate to retrieve the higher powers, thus retrieving $\beta$ parts of the file she requires. }%\textcolor{red}{Can probably write out the exact polynomial here.}}
	%   	$\mycode{C}\star\mycode{D}$ and $\mycode{C}\star{E}^f$, thus can decode $\mycode{C}\star{E}^f$, yielding the  parts of the file she requires.}
	%\awcomment{[Don't we want some details here how to "separate" and how to "decode"?]}
	%	We can assume that $\vec{E}^f$ is a higher power \awcomment{[What is a "higher power" Gabidulin code? Please try to write everything technical, not colloquial.]} Gabidulin code, thus, $C\star(D+\vec{E}^f)$ is again a Gabidulin code.
\end{remark}

We now assume that $2 \epsilon+\tau\leq d-1$, where $d$ is the minimum distance of the code $\mycode{C}\star(\mycode{D}+\mycode{E}^f)$. Then $\rk_{q^s}(N_e)=\epsilon$ errors and $\tau$ erasures can be corrected. In this case, the underlying PIR scheme would be similar to the scheme in \cite{tajeddine2018private}, with the ability to correct $\epsilon$ errors and $\tau$ erasures. For the rest of this section, we assume that $2 \epsilon+\tau\leq d-1$. In the current setting, the query code is constructed assuming a certain number of errors and erasures introduced by the network, thus, in this section, the probability of error is not taken into account.

%\cami{Make these rather into subsections than bullet points.}

We discuss the results for a single round $i$, and drop the superscript $(i)$. %The same result follow for other rounds. 
In the following, we distinguish two cases, namely the case where no errors occur on the uplink, and some errors occur on the downlink, and the case where some errors occur on the uplink as well.  {As opposed to the work done in \cite{wang2018pir, banawan2019capacity, banawan2018noisy}, we consider PIR over a random linear network, where the messages are coded over the servers, not replicated. Moreover, as opposed to all previous works on PIR, we consider uplink errors in this work.} 

\subsection{Downlink errors}
%\begin{itemize}
%\vspace{1.5em}
%\setlength\itemsep{1.5em}
%\item  
Let us first  assume that no errors are introduced on the uplink, but errors may be introduced on the downlink. In this case, it follows from \cite{silva2008rank} that the code can correct $\epsilon$ errors and $\tau$ erasures as long as $2\epsilon+\tau\leq d - 1$. %We can assume that we have no deviations introduced by the network. 
Since $\mycode{D}$ and $\mycode{E}^f$ are chosen such that $\mycode{C}\star(\mycode{D}+\mycode{E}^f)$ is a $\mycode{G}(n,k+t\rho+\beta-1)$ Gabidulin code, 
%we know the positions of the ``errors'' from the added matrix $\vec{E}^f$, we consider those erasures, and if we assume there are no erasures introduced by the network, thus we know that $\delta=\rank(\vec{E}^f)$. Therefore, for 
the user can decode and retrieve the requested file when $2\epsilon+\tau\leq d-1 = n-k-t\rho-\beta$, where $d$ is defined as the minimum distance of the code $\mycode{C}\star (\mycode{D}+\mycode{E}^{f})$. This can be summed up in the following theorem.%, where $\epsilon_r$ is the number of errors introduced by the network during the response phase.

\begin{theorem}\label{th:erdec}
	Assume that the queries are sent through a network where no errors occur, and the responses are sent through a network with errors. Moreover, assume that $$2\epsilon+\tau\leq d-1.$$
%where $\epsilon=\rank(N_e)$.
Then, the user can decode the errors and erasures introduced by the network, and consequently, can decode the requested file with rate $R_{\textup{\textsf{PIR}}}=\frac{\beta}{n}$.
\end{theorem}

\begin{proof}
This follows directly from \cite{silva2008rank}, where it was shown that if $\mycode{C}$ is a Gabidulin code with minimum distance $d$, then the errors introduced to the system will still allow for correct decodability when $$2\epsilon+\tau\leq d-1.$$

Since $\mycode{D}$ and $\mycode{E}^{f}$ %\awcomment{calligraphic? }
were chosen such that $(\Y\star (\vec{D}^\top+\vec{E}^{f}))$ is a codeword in a Gabidulin code with minimum distance $d\geq 2\epsilon+\tau+1$, thus the responses received by the user form a Gabidulin code that can correct up to $(d-1)/2$ errors.
\end{proof}

\subsection{Uplink errors}
%\item 
%The second case is more subtle. 
Assume that the queries are now sent through a network that introduces errors. In this case, the errors will be received by the servers. Since the dimension of the matrix sent to the servers is small enough to protect the user's privacy, the servers will have no capability to correct any errors. Hence, the servers will project the received matrix, including the errors, on their stored data and send it back to the user. We can see that an erasure that happens on the uplink would translate to an erasure on the downlink, so we can account for the total number of erasures as $\tau$.%We assume that no errors occur on the downlink, for the moment.

Consider an error $\vec{Z}_j$ of rank $r_Z$ %\awcomment{[I'd use a lower-case letter for denoting a rank]}
introduced on the channel going to server $j, j=1,\ldots, l$, such that a transformation matrix $\vec{B}_j$ is applied to $\vec{Z}_j$. Let $D_{Q_j}^{f}=\vec{D}_j^{\top}+\vec{E}^{f}_j$. Therefore, server $j$ receives the query matrix with additional errors:
\begin{align*}
	\vec{Q}^{f}_{j, rec} &= \left( \begin{array}{c:c}
\vec{A}_j+\vec{B}_j\vec{Z}_j & \vec{A}_jD_{Q_j}^{f}+\vec{B}_j\vec{Z}_j
\end{array}\right)\\
&= \left( \begin{array}{c:c}
\hat{\vec{A}_j} & \hat{Q_j}
\end{array}\right),
\end{align*}

where $\hat{Q}_j=\vec{A}_jD_{Q_j}^{f}+\vec{B}_j\vec{Z}_j$ is the erroneous received query, and $\hat{A}_j=\vec{A}_j+\vec{B}_j\vec{Z}_j$.

The server then projects the query matrix received on its stored data and then appends the matrix with an identity matrix. The server then returns to the user a double-lifted matrix \[\vec{R}^{f}_j = \left( \begin{array}{c:c:c}
I_{\rho\times \rho} &  \hat{\vec{A}}_j & \vec{Y}_j\star \hat{Q}_j
\end{array}\right),\]
where $\vec{Y}_j$ is the corresponding $j^{th}$ block of $\Y$.

The user receives the response \[
\vec{R}^{f}_{j,rec} = \left( \begin{array}{c:c:c}
\vec{A}^{\prime}_j+N_{e,j} & \vec{A}^{\prime}_j\hat{\vec{A}}_j+N_{e,j} & \vec{A}^{\prime}_j (\vec{Y}_j\star \hat{Q}_j)+N_{e,j}
\end{array}\right)\] from server $j$ transmitted through the network, where $N_{e,j}$ represents the errors introduced on the downlink from server $j$ to the user in round $i$.

From the received responses from all the servers, the user forms the matrix \begin{align*}\vec{R}^{f}_{rec}
% &=\small\left(\!\!\!\begin{array}{cccc:cccc:c}
% \vec{A}'_1+N_e \!\!\!& 0 & \cdots & 0 & \vec{A}'_1\hat{\vec{A}}_1+N_e \!\!\!& 0 & \cdots & 0 & \vec{A}'_1 (\vec{Y}_1\star \hat{Q}_1)+N_e\\
% 0 & \ddots & \ddots & \vdots & 0 & \ddots & \ddots & \vdots & \vdots \\
% \vdots & \ddots\!\!\! & \vec{A}'_{l-1}+N_e\!\!\! & 0 & \vdots & \ddots\!\!\! & \vec{A}'_{l-1}\hat{\vec{A}}_{l-1}+N_e\!\!\! & 0 & \vec{A}'_{l-1} (\vec{Y}_{l-1}\star \hat{Q}_{l-1})+N_e\\
% 0 & \cdots & 0 & \vec{A}'_l+N_e\!\!\! & 0 & \cdots & 0 & \vec{A}'_l\hat{\vec{A}}_l+N_e\!\!\! & \vec{A}'_l (\vec{Y}_l\star \hat{Q}_l)+N_e\\
% \end{array}\!\!\!\right)\\
&= \left( \begin{array}{c:c:c}
\vec{A}^{\prime}+N_e & \vec{A}^{\prime}\hat{\vec{A}}+N_e & \vec{A}^{\prime} (\vec{Y}\star \hat{Q})+N_e
\end{array}\right),\end{align*}
where 
$$\vec{A}^{\prime} = \left(\begin{array}{ccc}
    \vec{A}^{\prime}_1 & \cdots & 0\\
    \vdots & \ddots & \vdots\\
    0 & \cdots & \vec{A}^{\prime}_l\\
    \end{array}\right), \hat{\vec{A}} = \left(\begin{array}{ccc}
    \hat{\vec{A}}_1 & \cdots & 0\\
    \vdots & \ddots & \vdots\\
    0 & \cdots & \hat{\vec{A}}_l\\
    \end{array}\right),$$ \text{ and } $$N_e = \left(\begin{array}{ccc}
    N_{e,1} & \cdots & 0\\
    \vdots & \ddots & \vdots\\
    0 & \cdots & N_{e,l}\\
    \end{array}\right).$$

From the query construction, it is known that $\vec{Y}\star\vec{A}D_{Q_j}^{f}$ is a codeword of a rank metric code $\mycode{G}(n,k+t\rho+\beta-1)$, %\awcomment{[code or codeword? what are the parameters?]}, 
but the problem is that the received word, $\vec{Y}\star(\vec{A}D_{Q_j}^{f}+\vec{B}\vec{Z})$, is not a codeword of a rank metric code. However, $\vec{Y}\star(\vec{D}_{Q_j}^{f}+\vec{B}\vec{Z})$ can be seen as the codeword, $\vec{Y}\star D_{Q_j}^{f}$, of the rank metric code $\mycode{C}\star (\mycode{D}+\mycode{E}^f)$, %\awcomment{[calligraphic?]}
in addition to the noise introduced on the downlink, $N_e$, and some noise function $F(\vec{B}\vec{Z})$ of $\vec{A}D_{Q_j}^{f}$ along with the errors $\vec{B}\vec{Z}$. Let $\rk_{q^s}(F(\vec{B}\vec{Z}))\leq \epsilon_q$, and let $\epsilon = \epsilon_q+\epsilon_r$. We note here that $\epsilon$ is an upper bound to the rank of the errors introduced, and $\tau$ is an upper bound to the number of erasures in the system. % which has rank at most $min(rank(A(D+E_{f})^T),rank(BZ))$. We can assume that for decoding, we have $rank(BZ)<rank(A(D+E_{f})^T)$, thus $rank(F(BZ))\leq rank(BZ) = \epsilon_q$.

Now, if $n-k-t\rho-\beta \geq \tau+2\epsilon$, the user will be able to decode the requested file $X^f$ with a rate $\frac{\beta}{n}$.

%Now, assume $\rk_{q^s}(A^{\prime}(\vec{Y}\star\vec{A}D_{Q_j}^{f}) \geq k+2\epsilon$, the user will be able to decode the requested file $X^f$, with a rate $\frac{\rk_{q^s}(A^{\prime}\Y\star A\vec{E}^{f})}{n}$.

%\end{itemize}

%\vspace{1.5em}

In general, if errors occur on both links, we have the following theorem:

%\cami{[Rewrite in the same spirit as I rewrote the above thms. Assumptions. Claim. Proof. Be rigorous; which parts?]}
\begin{theorem}
	Assume that the queries and responses are sent through a network where errors occur. Moreover, assume that $n \geq \tau+2\epsilon+k+t\rho+\beta$. %$\rk_{q^s}(A^{\prime}(\vec{Y}\star\vec{A}D_{Q_j}^{f})) -2\epsilon\geq k$. %\awcomment{a bracket is missing}. 
The user can then decode the errors, and thus retrieve the requested file, $X^f$, with a rate $\geq\frac{n-k-t\rho-2\epsilon-\tau+1}{n}=\frac{\beta}{n}$.% the received response matrix.
    
    %\awcomment{Sorry, but I don't see how the number of errors that happens influences the rate??}
\end{theorem}

\begin{proof}
The proof of this theorem follows similarly to the proof of Theorem~\ref{th:erdec}. The inequality in the rate expression is due to the fact that full $t\rho$ collusion is assumed while only certain $t\rho$ sub-servers can collude.

%\cami{explain the inequality coming from the pessimistic full $t\rho$ collusion.}

\end{proof}

\section{Conclusion}
%\awcomment{[made some minor changes here, please check again]}
In this paper, we consider a distributed storage system such that the data is encoded using an MRD code and stored on servers. The user and servers communicate over a random \awcomment{linear} network. The user sends queries to the servers across the random network, such that the queries give no information about the file requested by the user to any server. It is further assumed that $t$ nodes collude in an effort to figure out the index of the file requested by the user, and the network is assumed to introduce erasures and errors into the queries and the responses. In this work, we have constructed a PIR scheme that allows the user to retrieve a file from a DSS over such a network privately. The scheme achieves asymptotically high PIR rates as $q\to\infty$.  {In this work, we consider the intermediate nodes in the network as dummy relays that only receive, combine and forward packets without storing data. However, this can be generalized in future work, as the intermediate nodes can be considered as other users that might collude or trade the information with third parties.}
%The objective is to allow the user obtain his/her requested file without revealing any information on the identity of the file to any of the nodes. We construct a PIR scheme against $t$ colluding nodes when the network has errors introduced, and when it doesn't. 

%\cami{Should we add more network coding papers in the refs, e.g. the one by Tracey Ho, Muriel Medard et al: A Random Linear Network Coding Approach to Multicast} \awcomment{I think, it's perfectly sufficient now.}
%\begin{example}
%
%\end{example}

\bibliographystyle{IEEEtran}
\bibliography{coding2}

\end{document}